%% file: varRobust.tex
\documentclass[11pt]{article}

\usepackage[inner=22mm,outer=21mm,tmargin=20mm,bmargin=20mm]{geometry}

\usepackage{latexsym,amsfonts,amsmath,amssymb,epsfig,tabularx,amsthm,dsfont,mathrsfs,mathtools}

\usepackage{graphicx}
\usepackage{enumerate}

\usepackage{titling}
\newcommand{\subtitle}[1]{%
  \posttitle{%
    \par\end{center}
    \begin{center}\large#1\end{center}
    \vskip0.5em}%
}

\usepackage{hyperref} 
\hypersetup{colorlinks=true,
        linkcolor=black,
        citecolor=black,
        urlcolor=blue}

\DeclareMathAlphabet{\mathup}{OT1}{\familydefault}{m}{n}

\newcommand{\widebar}[1]{\mbox{\kern1.5pt\hbox{\vbox{\hrule height 0.6pt \kern0.35ex
        \hbox{\kern-0.15em \ensuremath{#1 }\kern0.0em}}}}\kern-0.1pt}

\newlength{\fixboxwidth}
\usepackage{color}
\usepackage{soul}

\usepackage{latexsym,amsfonts,amsmath,amssymb,mathrsfs}
\usepackage{url}
\usepackage{graphicx}
\usepackage{color}
\usepackage{euscript}
\usepackage{verbatim}
\usepackage{hyperref}

\include{macros}

\newtheorem{corollary}{Corollary}
\newtheorem{alg}{Algorithm}
\newcommand{\Nv}{\mathrm N}
\renewcommand{\K}{\mathrm K}

\graphicspath{{Figures/}}

\newcommand{\gapm}[2]{{{\rm gap}_{#1}(#2)}}
\newcommand{\HS}{{\bbR^d}}
\newcommand{\mcH}{\mathcal B^d}
\newcommand{\LAn}{\mathcal L_{\pi_n}}

\renewcommand{\S}{\mathrm S}
\newcommand{\PM}{\mathrm{P}}

\begin{document}

\title{Robust random walk-like Metropolis--Hastings algorithms for concentrating posteriors}
\author{Daniel Rudolf\thanks{Universit\"at Passau, Innstraße 33, 94032 Passau, Germany, Email: daniel.rudolf@uni-passau.de}, Bj\"orn Sprungk\thanks{Faculty of Mathematics and Computer Science, Technische Universit\"at Bergakademie Freiberg, Pr\"uferstr. 9, 09599 Freiberg, Email: bjoern.sprungk@math.tu-freiberg.de}
}

\date{\today}
\maketitle

\begin{abstract}
Motivated by Bayesian inference with highly informative data we analyze the performance of random walk-like Metropolis--Hastings algorithms for approximate sampling of increasingly concentrating target distributions.
We focus on Gaussian proposals which use a Hessian-based approximation of the target covariance.
By means of pushforward transition kernels we show that for Gaussian target measures the spectral gap of the corresponding Metropolis--Hastings algorithm is independent of the concentration of the posterior, i.e., the noise level in the observational data that is used for Bayesian inference.
Moreover, by exploiting the convergence of the concentrating posteriors to their Laplace approximation we extend the analysis to non-Gaussian target measures which either concentrate around a single point or along a linear manifold.
In particular, in that setting we show that the average acceptance rate as well as the expected squared jump distance of suitable Metropolis--Hastings Markov chains do not deteriorate as the target concentrates.
\end{abstract}

\noindent{\bf Classification. Primary: 65C40; Secondary: 60J22, 62D99, 65C05.} 
\\[1ex]


\section{Introduction}
The challenging goal to gain knowledge from distributions of interest by approximate sampling is omnipresent in computational statistics. For example, in Bayesian inference the prediction based on sampling posterior distributions is crucial or in statistical physics one draws conclusions from Gibbs measures and their samples. In particular, being able to generate a realization of a target  probability measure efficiently leads directly to a proxy of expectations of quantities of interest
by taking the mean w.r.t. the empirical distribution. Unfortunately, exact sampling is in general infeasible (because of unknown normalization constants or other computational issues), so that a standard approach via Markov chains is commonly used. In particular, the ability of simulating Markov chains that converge marginally to the target distribution is heavily exploited. 

Still the most prominent methodology for realizing such Markov chains is given by the Metropolis-Hastings (MH) algorithm. For target distribution $\pi$ and proposal kernel $P$, defined on $\HS$, a transition from $x\in\HS$ of the MH algorithms works (essentially) as follows: Realize a sample w.r.t. $P(x,\cdot)$, call it $y$, and return this new state with MH acceptance probability $\alpha_{P}(x,y)$ that depends on $\pi$ and $P$ and otherwise return $x$. 
How fast the distribution of $X_n$ of the corresponding Markov chain $(X_n)_{n\in\mathbb{N}}$ converges to $\pi$ (for $n\to\infty$) depends heavily on the choice of the proposal kernel. Gaussian random walk proposals $P(x,\cdot) = \Nv(x, s^2 C)$ appear to be commonly used in practice with the advantage that they allow the tuning of specific parameters such as the so-called stepsize $s>0$ or the proposal covariance $C \in \bbR^{d\times d}$. 

The tuning should improve the performance of the MH algorithm, i.e., yield a faster convergence of the associated Markov chain to its limit distribution or a higher effective sample size of the corresponding classical Markov chain Monte Carlo (MCMC) estimator for the approximation of expectations of quantities of interest.
Both, the speed of convergence as well as the effective sample size, can be controlled by the spectral gap of the Markov operator associated to the Markov chain.
Moreover, two other common measures for tuning and studying the performance of MH algorithms are the averaged acceptance rate and the expected squared jump distance of the corresponding Markov chain. 

How these efficiency quantities behave for an increasing state space dimension $d$ and how to optimally tune or scale, e.g., the stepsize parameter $s$ with respect to $d$, has been intensively studied in the past decades, see e.g., \cite{RobertsRosenthal2001}.
Moreover, in recent years modifications of classical MH algorithms have been developed which are well-posed in infinite dimensional state spaces and, thus, show a dimension-independent efficiency, see, e.g., \cite{BeskosEtAl2008, BeskosEtAl2011, CotterEtAl2013, HaStVo14, NortonFox2016, RudolfSprungk2018}.
In particular, in \cite{HaStVo14, RudolfSprungk2018} the dimension-independence of the spectral gap of MH algorithms based on the so called (generalized) preconditioned Crank-Nicolson proposal, a random walk-like Gaussian proposal, has been proven.
In addition, in \cite{NortonFox2016} conditions for a dimension-independent averaged acceptance rate and expected squared jump distance have been derived for a wide class of proposal kernels.

However, not only the dimension $d$ of the state space $\HS$ affects the efficiency of MCMC methods and requires suitable scaling of tuning parameters.
Also the concentration of the target measures, i.e., how widespread or focused the distribution $\pi$ is in various directions of $\HS$, can significantly influence the performance of MH algorithms.
Despite its importance for Bayesian inference in practice, surprisingly, the problem of highly concentrated or ridged target measures has drawn rather less attention in the MCMC literature (with notable exceptions \cite{AuEtAl2020, BeskosEtAl2018}). Let us comment on the aforementioned importance.
In many applications, for instance, in subsurface geophysics, observational data is rather sparse but the corresponding signal-to-noise ratio quite large.
Thus, the data is highly informative for certain directions in the parameter space $\HS$, but others remain rather unaffected by the data or likelihood, respectively. 
Hence, the resulting posterior distribution, in a Bayesian inference framework based on this data, is highly concentrated in specific directions.
This high concentration is of advantage from an inference point of view, since it represents only little remaining uncertainty about the unknown parameter $x$. 
However, it poses a serious challenge for an effective exploration of the posterior distribution by MH algorithms, since the high concentration allows only  `small steps' of classical random walk-like MH Markov chains.

In order to analyse the performance of MH algorithms for approximate sampling of posteriors resulting from informative data, we consider a sequence of increasingly concentrating target probability measures $(\pi_n)_{n\in\mathbb{N}}$. Each distribution $\pi_n$ is defined on $\HS$ and given by
\[
	\pi_n(\d x)
	\coloneqq
	\frac 1{Z_n} \exp(-n U(x))\ \pi_0(\d x),
	\qquad
	Z_n \coloneqq \int_\HS \exp(-n U(x))\ \pi_0(\d x),
\]
with $\pi_0$ denoting a reference probability measure, $U\colon \HS\to [0,\infty)$ a negative log-likelihood or potential, and $Z_n$ the normalizing constant of $\pi_n$. 
Such target measures occur, for instance, in Bayesian inference where $X\sim \pi_0$ is inferred based on a realization of an observable $Y = F(X) + n^{-1/2} \varepsilon$ with forward mapping $F$ and Gaussian observational noise $\varepsilon \sim \Nv(0, \Sigma)$.
The level of `information' or concentration is thus controlled by $n\in\bbN$ and as $n\to \infty$ the distribution $\pi_n$ concentrates around the set $\mc M_U \coloneqq \argmin_{x \in \supp \pi_0} U(x)$, where $\supp \pi_0$ denotes the support of $\pi_0$.
Now, given that the concentration of $\pi_n$ naturally restricts the `jump size' of MH Markov chains, a decreasing scaling of the stepsize $s = s(n)$ with $n$ is expected to be reasonable. In fact, the optimal scaling of stepsize parameters in a setting of isotropic Gaussian random walk proposals $P_n(x,\cdot) = \Nv(x, s(n)^2 I)$ has been investigated in \cite{BeskosEtAl2018}. 

However, a simple scaling of the parameter $s$ is in general not sufficient in order to guarantee a robust performance.
Particularly, in the common case that not every direction in the parameter space is informed by the likelihood, an isotropic scaling of the stepsize yields a slower and slower exploration in the uninformed directions, which in turn leads to a decreasing spectral gap.
Thus, concentration robust MH algorithms need to be based on proposal kernels which exploit some information about the (local) geometry or concentration, respectively, of the target distribution.

A common and simple idea, to include such information, is to use an approximation of the target covariance within the proposal kernel.
This dates back at least to \cite{Tierney1994} and has led, for instance, to the development of the adaptive Metropolis algorithm in \cite{HaarioEtAl2001}, where the target covariance is estimated by the empirical covariance of 
(previous) realizations of 
the Markov chain.
In recent years, approximations based on gradient and Hessian information of the posterior density have been exploited for the construction of Gaussian proposal kernels, see \cite{MartinEtAl2012, Law2014, CuiEtAl2014, RudolfSprungk2018}, which relate to the classical \emph{Laplace approximation} of posterior measures in Bayesian statistics.
These covariance approximations can usually be computed offline, i.e., before running the MH algorithm, by numerical optimization methods. 
In particular, in \cite{RudolfSprungk2018} the authors observed in numerical experiments that Gaussian random walk-like proposal kernels using a gradient-based approximation to the target covariance showed a quite robust effective sample size for an increasing concentration of the posterior.
Moreover, the authors of \cite{AuEtAl2020} obtained similar numerical results for a novel Hamiltonian Monte Carlo sampler. \\[-2ex]

\textbf{Contributions~}
We provide a first theoretical analysis of the performance of MH algorithms using Hessian-based approximations $C_n$ of the covariance of the target $\pi_n$ 
within the proposal kernel.
We focus on two Gaussian random walk-like proposals: 
\[
P_n(x,\cdot) = \Nv(x, s^2 C_n) \quad  \text{and} \quad P_n(x,\cdot) = \Nv(x_n + \sqrt{1-s^2}(x-x_n), s^2 C_n),
\]
with
\[
	x_n \coloneqq \argmin_{x\in \supp \pi_0} U(x) - \frac 1n \log \pi_0(x),
	\qquad
	C_n \coloneqq \frac 1n \left(\nabla^2 U(x_n) - \frac 1n \nabla^2 \log \pi_0(x_n) \right)^{-1},
\]
where $\pi_0$ (by an abuse of notation) denotes also the Lebesgue density of the reference distribution $\pi_0$ and where $x_n$ is assumed to be the unique, nondegenerate minimiser.
Note that $x_n$ denotes the maximum a-posteriori estimate, i.e., the maximiser of the posterior (Lebesgue) density.

The latter of the two proposals is a `modified' version of the well-known preconditioned Crank-Nicolson proposal \cite{CotterEtAl2013}, which was shown to yield a dimension-independent spectral gap under suitable conditions \cite{HaStVo14}.
We want to make explicitly clear, that 
no additional scaling of the stepsize parameter $s$ is required.
The use of the Hessian-based approximating covariance $C_n$ contains an implicit scaling of the `jump size' in exactly those directions which are affected by the likelihood or $U$, respectively.
For these proposal kernels we show the following two main results:
\begin{enumerate}
\item[\textbf{1.}] \textbf{Gaussian target result:}
In the case of a Gaussian target measures $\pi_n$, the MH algorithm based on one of the above proposal kernels $P_n$ yields a spectral gap which is independent of the concentration level $n$.
Similarly, also the average acceptance probability and the (with respect to the target variance) normalized expected jump squared distance of the associated Markov chain are independent of $n$. 

\item[\textbf{2.}] \textbf{Non-Gaussian target result:}
If the non-Gaussian target measures $(\pi_n)_{n\in\mathbb{N}}$ converge in Hellinger distance to their \emph{Laplace approximations}, which are given by the Gaussian measures $\LAn \coloneqq \Nv(x_n, C_n)$, then the MH algorithm based on one of the above proposal kernels $P_n$ has a non-deteriorating average acceptance rate and a non-deteriorating normalized expected jump squared distance as the concentration level of $\pi_n$ increases, i.e., as  $n\to\infty$. 
\end{enumerate}
The first main result follows by a straightforward application of the concept of \emph{pushforward transition kernels} which we outline in detail in the appendix. The second main result is based on a stability analysis of the corresponding efficiency quantities with respect to perturbations of the target measure combined with recent results on the convergence of the Laplace approximation from \cite{SchillingsEtAl2020}.\\[-2ex]

\textbf{Implications and Limitations~}
Besides a rigorous analysis of MH algorithms for concentrating posteriors, our theoretical results also provide a solid justification of the general advice stated in \cite{Tierney1994} to use the target covariance (or approximations to it) for proposing new states in MH algorithms, i.e., then the performance, for instance the spectral gap for Gaussian targets, is completely independent of the concentration level of the posterior.

However, in order to have convergence of $\pi_n$ and $\LAn$ to each other in Hellinger distance, the posteriors $\pi_n$ must either converge (weakly) to a point mass or concentrate around a linear manifold $\mc M_U$ where $\pi_0$ must also be Gaussian in the latter case.
Thus, for the case that not all directions of the parameter space are affected by the observational data or likelihood, respectively, our analysis holds only if the negative log likelihood $U$ and, thus, the observational data depends only on directions belonging to an \emph{active subspace}, cf. \cite{Constantine2015, CuiEtAl2014}.
A similar assumption is required in the optimal scaling results in \cite{BeskosEtAl2018}.
We see our work as a starting point for the analysis of Hessian-based MH algorithms for concentrated posteriors.
Future work will extend the results presented here to more general settings where $\pi_n$ (as $n\to \infty$) concentrates along nonlinear manifolds.
Then, a single Hessian will most likely not be sufficient for a concentration robust performance, rather local Hessian information will be required.\\[-2ex]
\textbf{Outline~}
In Section~\ref{sec:preliminaries} we state the general setting and notation as well as the required definitions for MH algorithms and measuring their performance.
Section~\ref{sec:VarRobust_Intro} introduces our concepts for concentration robust performance as well as the considered Hessian-based proposal kernels.
Our theoretical analysis of the resulting MH algorithms is then conducted in Section~\ref{sec:VarRobust_Gauss}.
In particular, our first main result is shown in Section~\ref{sec:Gaussian_target} and the second main result in Section~\ref{sec:NonGaussian_target}.
In the appendix we state the results on pushforward transition kernels required in Section~\ref{sec:Gaussian_target}.


\section{Preliminaries}\label{sec:preliminaries}
We assume to have a sufficiently rich probability space $(\Omega,\mathcal{A},\mathbb{P})$ that serves as common domain for all subsequently appearing random variables and denote by $\mcH$ the Borel $\sigma$-algebra of $\HS$. Let $\mc P(\HS)$ be the set of all probability measures on $(\HS,\mcH)$. Now consider the problem of sampling w.r.t. probability measures $\pi\in \mc P(\HS)$ of the form
\begin{equation}\label{equ:post0}
	\pi(\d x)
	\coloneqq
	\frac 1{Z} \exp(- U(x))\, \pi_0(\d x),
	\qquad
	Z \coloneqq \int_\HS \exp(-U(x))\, \pi_0(\d x),
\end{equation}
where $U \colon \HS \to [0,\infty)$ is a measurable function and $\pi_0\in\mathcal{P}(\HS)$ a reference measure.
Our motivation for such types of target measures comes from Bayesian inference in $\HS$ given noisy observations $y \in \bbR^m$ of an observable of the form
\[
    Y = F(X) + \varepsilon,
    \qquad
    X \sim \pi_0,
    \qquad
    \varepsilon \sim \pi_\text{noise}.
\]
Here $F\colon \HS\to\bbR^m$ denotes a forward map or mathematical model, for which the parameter $x \in \HS$ needs to be inferred, $\pi_0$ represents the prior probability measure, and $\varepsilon$ denotes an observational noise which is stochastically independent from $X\sim \pi_0$. 
Often in practice, the random variable $\varepsilon$ follows a mean-zero Gaussian  distribution $\pi_\text{noise} = \Nv(0, \Sigma)$ with regular covariance matrix $\Sigma \in \bbR^{m\times m}$.
Then, the resulting posterior distribution $\pi$ of $X\sim \pi_0$ given $Y=y$ is of the form \eqref{equ:post0} with $U(x) \coloneqq \frac 12 \|y - F(x)\|^2_{\Sigma}$ where $\|v\|_\Sigma \coloneqq \sqrt{v^\top \Sigma^{-1} v}$ for $v\in\bbR^m$.

Now one is interested to sample (approximately) from $\pi$ to gain knowledge about the posterior distribution. 
For example, a common task is to compute posterior expectations $\pi(f)$ of quantities of interest $f\colon \HS \to \bbR$ given by
\[
    \pi(f) \coloneqq \int_{\HS} f(x)\ \pi(\d x).
\]
For these purposes one constructs a `suitable' Markov chain with transition kernel\footnote{A mapping $K\colon \bbR^d\times \mcH \to [0,1]$ is called transition kernel if $K(\cdot,A)\colon \HS \to [0,1]$ is, for any $A\in\mcH$, a measurable function and $K(x,\cdot)\colon\mcH\to [0,1]$ is, for any $x\in\HS$, a probability measure.} $K$, that is, a $\bbR^d$-valued sequence of random variables $(X_k)_{k\in\bbN}$ satisfying
\[
    \mathbb{P}(X_{k+1}\in A\mid X_{1},\dots,X_k) = \mathbb{P} (X_{k+1}\in A\mid X_k) = K(X_k,A) \qquad k\in\bbN, \, A\in\mcH.
\]
Here `suitable' means that for sufficiently large and further increasing $k\geq 1$ the distribution of $X_k$ `gets close' to $\pi$. Moreover, the (standard) Markov chain Monte Carlo (MCMC) estimator $S_N(f) \coloneqq \frac{1}{N} \sum_{k=1}^N f(X_k)$ approximates $\pi(f)$ consistently under weak assumptions 
\cite[Chapter~17]{MeynTweedie2009}. A minimal requirement for the aforementioned convergence statements is that $\pi$ is a stationary distribution w.r.t. the corresponding transition kernel $K$, that is,
\[
    \pi(A) = \pi K(A) \coloneqq \int_{\HS} K(x,A)\ \pi(\d x), \quad A\in\mcH.
\]
A more demanding property is that $K$ is reversible w.r.t. $\pi$, i.e., $K(x,\d y)\pi(\d x) = K(y,\d x) \pi(\d y)$. This symmetry condition of the measure $K(x,\d y)\pi(\d x)$ on $\HS\times \HS$ implies that $\pi$ is a stationary distribution.  

We focus on Markov chains based on the Metropolis--Hastings algorithms with commonly used random-walk like proposals that lead to reversible transition kernels.


\paragraph{Metropolis--Hastings algorithms}
Suppose that we have a target measure $\pi \in \mc P(\HS)$ of the form \eqref{equ:post0}, a proposal transition kernel $P$ on $\HS$ as well as a function $\alpha_P \colon \HS\times \HS \to [0,1]$, which serves as acceptance probability that might depend on $\pi$ and $P$. Then, the most prominent methodology for the construction of a transition kernel that is reversible w.r.t. $\pi$ is given by the Metropolis-Hastings (MH) algorithm.
It proceeds as follows.
\begin{alg}(Metropolis-Hastings).
Given a current state $x\in\HS$, one obtains the next state $x_\text{next} \in\HS$ by the following steps:
\begin{enumerate}
    \item Draw $Y\sim P(x,\cdot)$ and $U\sim \text{Unif}[0,1]$ independently and denote the realisations by $y$ and $u$ respectively.
    \item If $u<\alpha_{P}(x,y)$, return $x_{\text{next}}\coloneqq y$, otherwise return $x_\text{next}\coloneqq x$.
\end{enumerate}
\end{alg}
The algorithm can be rewritten in terms of a transition kernel:
\[
    K(x, \d y)
    =
    \alpha_P(x,y) \, P(x, \d y) + r(x) \, \delta_x(\d y),
    \qquad
    r(x) \coloneqq \int_\HS (1- \alpha_P(x,y)) \, P(x, \d y).
\]
Here $\delta_x$ denotes the Dirac measure located at $x\in\HS$ and the function $r$ is called the `rejection probability'. 
%
Now we specify the acceptance probability $\alpha_P$ of the MH algorithm that eventually implies the well known fact that the transition kernel $K$ is reversible w.r.t. $\pi$. It is given by 
\begin{equation}\label{eq:alpha}
	\alpha_P(x,y) 
	\coloneqq
	\min\left\{1, \exp\left(U(x) - U(y) \right)\, h_P(x,y) \right\},
\end{equation}
where $h_P(x,y) \coloneqq \frac{\d \nu_P^\top}{\d \nu_P}(x,y)$ with $\nu_P(\d x \d y) \coloneqq P(x,\d y) \pi_0(\d x)$ and $\nu^\top_P(\d x \d y) \coloneqq \nu_P(\d y \d x)$, see \cite{Ti98}. 

Assume that $P$ admits a Lebesgue density $p\colon \HS\times \HS \to [0,\infty)$ such that $P(x,A) = \int_A p(x,y)\ \d y$ for any $x\in\HS, A\in\mcH$. 
Additionally suppose that $\pi_0$ possesses a Lebesgue density which we also denote by $\pi_0\colon \HS\to[0,\infty)$ and let \[
    \S_0  \coloneqq \supp \pi_0 \coloneqq \{ x\in\HS \colon \pi_0(x)>0 \}.
\]
Then, by \cite{Ti98} we have
\[
  h_P(x,y) \coloneqq
  \begin{cases}
  \frac{\pi_0(y) \ p(y,x)}{\pi_0(x)\ p(x,y)}, & \frac{p(x,y)}{\pi_0(y)}>0 \quad \text{and} \quad \frac{p(y,x)}{\pi_0(x)}>0,\\ 
  0, & \text{otherwise}.
  \end{cases}
\]
We focus on rather simple proposal transition kernels, such as
\begin{enumerate}
\item
the Gaussian random walk proposal,
\[
	P(x,\cdot) = \Nv\left(x, s^2 C \right), \quad s>0,
\]
with step size parameter $s>0$ and covariance matrix $C\in\bbR^{d\times d}$, leading to 
\[
    \alpha_P(x,y) 
    =
	\min\left\{1, \exp\left(U(x) - U(y) \right)\, \frac{\pi_0(y)}{\pi_0(x)} \,\boldsymbol{1}_{\S_0\times \S_0}(x,y) \right\};
\]
\item
the  preconditioned Crank--Nicolson (pCN) proposal introduced by \cite{Neal1999, BeskosEtAl2008}:
\[
	P(x,\cdot) = \Nv\left(\sqrt{1-s^2} x, s^2 C \right), \quad s\in(0,1],
\]
which is reversible w.r.t.~$\Nv(0,C)$. Let us denote the Lebesgue density of $\Nv(0,C)$ by $\varphi_{0,C}$. 
Then, the corresponding acceptance probability satisfies
\[
    \alpha_P(x,y) 
    = \min\left\{1, \exp\left(U(x) - U(y) \right)\, \frac{\pi_0(y)}{\pi_0(x)}\, \frac{\varphi_{0,C}(x)}{\varphi_{0,C}(y)} \,\boldsymbol{1}_{\S_0\times \S_0}(x,y)  \right\}.
\]
Note that for $\pi_0=\Nv(0,C)$ the MH algorithm is even well-defined in infinite dimensional Hilbert spaces, since actually no Lebesgue densities for the definition of the acceptance probability are required. Under suitable assumptions on $U$ the resulting MH algorithm (where $\pi_0=\Nv(0,C)$) has dimension-independent statistical efficiency in terms of a dimension-independent spectral gap \cite{HaStVo14}.
\end{enumerate}


\paragraph{Measures for performance}
In order to study the performance of MH algorithms for increasingly concentrated posteriors distributions we require specific measures of efficiency. Those quantify the performance of the Markov chain methodology by real numbers that we can analyze as the concentration of the target distribution increases. 

We require some notation. Let $L^p_\pi$ for $p\geq 1$ be the Lebesgue space of measurable functions $f\colon \HS \to \bbR$ such that $\pi(|f|^p) < \infty$. Further, we set $\|f\|_{L^p_\pi} \coloneqq \left( \pi(|f|^p) \right)^{1/p}$ for $f \in L^p_\pi$ and denote the inner product in $L^2_\pi$ by
\[
    \langle f, g \rangle_{\pi} \coloneqq \int_{\HS} f(x) g(x) \, \pi(\d x),
    \qquad
    f,g\in L^2_\pi.
\]
Since MCMC methods are used to compute posterior expectations of quantities of interest $f\colon \HS\to\bbR$, one is interested in the asymptotic variance of $S_N(f) = \frac 1N \sum_{k=1}^N f(X_k)$ where $(X_k)_{k\in\bbN}$ denotes a $\pi$-reversible  Markov chain realized by a MH algorithm. 
Given that the Markov chain is irreducible, $X_1\sim \pi$ and $\Var_\pi(f) \coloneqq \pi\left( (f - \pi(f))^2 \right)$, the asymptotic variance satisfies
\[
    \lim_{N\to\infty} N\ \ev{\left|S_N(f) - \pi(f)\right|^2}
    =
    \Var_\pi(f)\left(1 +  2 \sum_{k=1}^\infty \Corr\left(f(X_1), f(X_{k+1})  \right)\right),
\]
see \cite{KipnisVaradhan1986}. Thus, for a given quantity of interest $f\in L_\pi^2$ a common measure of efficiency is the associated \emph{integrated autocorrelation time}
\[
    \tau_f \coloneqq 1 +  2 \sum_{k=1}^\infty \Corr\left(f(X_1), f(X_{k+1})\right), \qquad X_1 \sim \pi.
\]
This number can be rewritten in terms of the Markov operator $\K\colon L^2_\pi \to L^2_\pi$ that is associated to the transition kernel $K$ and given by $\K f(x) \coloneqq \int_{\HS}f(y)\ K(x,\d y)$. 
To this end, define the `standardized' quantity of interest $\widebar f_\pi(x) \coloneqq \frac{f(x) - \pi(f)}{\sqrt{\Var_\pi(f)}}$, i.e., $\pi(\widebar f_\pi) = 0$ and $\Var_\pi(\widebar f_\pi)=1$. Then we have
\begin{equation}\label{eq:iart_K}
    \tau_f 
    =
    \tau_f(K)
    \coloneqq 1 + 2 \sum_{k=1}^\infty \langle \widebar f_\pi, \K^k \widebar f_\pi\rangle_\pi.
\end{equation}
This measure of efficiency obviously depends on the choice of $f \in  L^2_\pi$. 
Related to the integrated autocorrelation time, but independent of a quantity of interest $f$, is the \emph{spectral gap} of $\K$,
\begin{equation}\label{eq:gap_K}
    \gapm{\pi}{K} 
    \coloneqq
    1 - \sup_{0\neq f \in L^2_\pi} \frac{\|\K f - \pi(f)\|_{L^2_\pi}}{\|f\|_{L^2_\pi}}.
\end{equation}
The autocorrelation time and the spectral gap satisfy the following relation
\begin{equation}\label{eq:SG_IAT}
    \tau_f(K) \leq \frac{2 \|f\|_{L^2_\pi}}{1 - \gapm{\pi}{K}} \qquad
    \forall f \in  L^2_\pi,
\end{equation}
see for example \cite{Rudolf2012}. 

We now focus on two simpler efficiency quantities which will be our main focus in the subsequent sections. The first one is the \emph{(expected squared) jump distance}
\[
	\rho(K)
	\coloneqq
    \left( \int_{\HS} \int_{\HS} \|x- y\|^2\ K(x,\d y)\ \pi(\d x) \right)^{1/2},
\]
that can be refined via the \emph{directional (expected squared) jump distance} 
\begin{equation*}
	\rho_v(K)
	\coloneqq
    \left( \int_{\HS} \int_{\HS} |v^\top x - v^\top y|^2\ K(x,\d y)\ \pi(\d x) \right)^{1/2}
\end{equation*}
where $v\in\HS$ with $\|v\|=1$ denotes a considered direction. Intuitively, 
the number $\rho_v(K)$ tells us how well a Markov chain with transition kernel $K$ can explore the whole $\HS$ in the direction $v$.  
Moreover, the lag one autocorrelation for linear quantities of interest $f_v(x) \coloneqq v^\top x$ is closely related to $\rho_v(K)$.
To see this, let $(X_k)_{k\in\bbN}$ be Markov chain with transition kernel $K$ and stationary initial distribution. Then, by means of $\Var(X_{k+1} -  X_{k}) = 2 \Var(X_{k}) - 2\Cov(X_{k},X_{k+1})$ we have
\[
	\Corr\left( v^\top X_k, v^\top X_{k+1} \right)	
	=
	1
	- \frac{\ev{ \left|v^\top\left(X_{k+1} - X_{k}\right)\right|^2 }}{2\Var\left(v^\top X_k\right)}
	=
	1 - \frac{\rho_v(K)}{2\Var_{\pi}(f_v)}.
\]
Motivated by that we introduce the \emph{normalized directional (expected squared) jump distance} 
\begin{equation}\label{equ:SJD_def}
	\widebar \rho_v(K)
	\coloneqq
	\frac{\rho_v(K)}{\Var_{\pi}(f_v)},
	\qquad
	f_v(x) = v^\top x, \ \ 0\neq v \in \bbR^d.
\end{equation}
Intuitively $\widebar \rho_v(K)$ describes the expected squared jump size of the Markov chain in direction $v$ relatively to the `spread' of the target measure in this direction. In particular, since correlations are bounded by $1$, we obtain that $0\leq \widebar \rho_v(K) \leq 2$.
Moreover, if the Markov operator $\K \colon L_\pi^2 \to L_\pi^2$ associated to $K$ is positive semi-definite, then we have $\Corr\left( v^\top X_k, v^\top X_{k+1} \right) \geq 0$ and, thus, $0\leq \widebar \rho_v(K) \leq 1$.
In this case, we also have $\widebar \rho_v(K) \leq \frac 12 \left(\tau_{f_v}(K) - 1\right)$.

Eventually, we consider the (averaged) acceptance rate of a $\pi$-reversible MH transition kernel $K$ with proposal kernel $P$ given by
\begin{equation}\label{equ:AAR_def}
	\widebar{\alpha}(K)
	\coloneqq
	\int_{\HS} \int_{\HS} \alpha_P(x, y) \ P(x,\d y)\ \pi(\d x),
\end{equation}
with $\alpha_P$ as in \eqref{eq:alpha}. 
The averaged acceptance rate is typically used for tuning MH algorithms in practice.
For instance, for the Gaussian random walk proposal $P(x,\cdot) = \Nv\left(x, s^2 C \right)$ the usual `rule of thumb' is to choose a step size $s>0$ such that $\widebar{\alpha}(K) \approx 0.234$ \cite{RobertsRosenthal2001}. Usually this yields a `good' scaling w.r.t. an increasing state space dimension $d\to\infty$ and provides also `small' asymptotic variances.
Although, the averaged acceptance rate seems to be the crudest and simplest quantity for measuring the efficiency among the four presented ones, it serves as a bound on the maybe most sophisticated one, the spectral gap: By means of the conductance of the underlying Markov chain and Cheeger's inequality \cite{LawlerSokal1988} we indeed have
\begin{equation}\label{eq:gap_AAR}
    \gapm{\pi}{K}
    \leq
    2\widebar{\alpha}(K).
\end{equation}


\section{Concepts and Kernels for Concentration Robustness}\label{sec:VarRobust_Intro}
In order to study the performance of MCMC methods for highly concentrated posterior measures resulting from highly informative data, we consider a sequence of increasingly concentrated target probability measure $\pi_n \in \mc P(\HS)$, $n\in\bbN$, on $\HS$ given by
\begin{equation}\label{equ:post}
	\pi_n(\d x)
	\coloneqq
	\frac 1{Z_n} \exp(-n U(x))\ \pi_0(\d x),
	\qquad
	Z_n \coloneqq \int_\HS \exp(-n U(x))\ \pi_0(\d x),
\end{equation}
with $\pi_0$ denoting the prior or other suitable reference measure on $\mcH$, $U\colon \HS\to [0,\infty)$ denoting a negative log-likelihood or potential, and $Z_n$ being the normalizing constant of $\pi_n$.
For $U(x) \coloneqq \frac 12 \|y - F(x)\|^2_{\Sigma}$ this corresponds to the above Bayesian inference problem with scaled noise
\[
    Y = F(X) + \varepsilon, \qquad X\sim \pi_0, \quad
    \varepsilon \sim \Nv\left(0, \frac 1n \Sigma\right), \ n\in\bbN.
\]

\begin{rem}[On the assumption $U\geq 0$]
Regarding the standing assumtpion that $U(x) \geq 0$ for $U$ in \eqref{equ:post} we remark this is a rather mild restriction.
If $U$ were to be unbounded from below, i.e., $\essinf_{x\in\HS} U(x) = -\infty$, then the sequence of measures $\{\pi_n\}_{n\in\bbN}$ is not tight \cite[Proposition 2.1]{Hwang1980}. 
In order to prevent that we, hence, have to assume that $\essinf_{x\in\HS} U(x) > -\infty$ which w.~l.~o.~g.~can be restated as $\essinf_{x\in\HS} U(x) = 0$ using a suitable normalization constant $Z_n$.
\end{rem}


\subsection{Notions of Concentration Robust Performance}
We now consider a family of transition kernels $\{K_n\}_{n\in\bbN}$ where each transition kernel $K_n$ corresponds to a MH algorithm targeting $\pi_n$ as in \eqref{equ:post} and employing a proposal kernel $P_n\colon \HS\times \mcH \to[0,1]$ and an acceptance probability $\alpha_n\colon \HS\times\HS\to[0,1]$ according to \eqref{eq:alpha}.
Thus, $K_n$ is $\pi_n$-reversible.
In the following we define and study notions of a robust performance of $K_n$ as $n\to\infty$, i.e., as the target measure $\pi_n$ becomes more concentrated.
To this end, we use the four measures or quantities, respectively, of efficiency discussed in Section \ref{sec:preliminaries}.

\begin{defi}
We say, a family $\{K_n\}_{n\in\bbN}$ of $\pi_n$-reversible MH transition kernels $K_n$ as assumed above has 
\begin{itemize}
    \item 
    a \emph{(concentration) robust average acceptance rate} if
    \begin{equation}\label{equ:AAR_robust} \tag{RAAR}
	    \liminf_{n\to\infty} \widebar{\alpha}(K_n) > 0;
    \end{equation}
    
    \item
    a \emph{(concentration) robust jump squared distance} if
    \begin{equation}\label{equ:RJD_robust}\tag{RSJD}
    	\liminf_{n\to\infty} \widebar{\rho}_v(K_n) > 0
    	\qquad
    	\forall 0 \neq v\in \bbR^d;
    \end{equation}
    
    \item
    a \emph{(concentration) robust integrated autocorrelation time} if
    \begin{equation}\label{equ:IACT_robust}\tag{RIAT}
    	\limsup_{n\to\infty} \tau_f(K_n) < \infty
    	\qquad
    	\forall f \in L^2_{\pi_0}(\bbR);
    \end{equation}

    \item
    a \emph{(concentration) robust spectral gap} if
    \begin{equation}\label{equ:Gap_robust}\tag{RSG}
    	\liminf_{n\to\infty} \gapm{\pi_n}{K_n} > 0.
    \end{equation}
\end{itemize}
\end{defi}

The basic condition \eqref{equ:AAR_robust} is essential for all other, more ambitious, notions of concentration robust performance. 
In particular, due to the relation \eqref{eq:gap_AAR} we know that if \eqref{equ:AAR_robust} does not hold, then \eqref{equ:Gap_robust} can not hold either.
Similarly, if \eqref{equ:RJD_robust} does not hold, then obviously also \eqref{equ:IACT_robust} can not hold.
Moreover, it can be shown that \eqref{equ:Gap_robust} also implies \eqref{equ:RJD_robust}, see \cite[Proposition 6.3]{Sprungk2017}.
Lastly, by \eqref{eq:SG_IAT} we also know that \eqref{equ:Gap_robust} implies \eqref{equ:IACT_robust}.
Thus, we can summarize the relation between the four notions of concentration robust performance as follows:
\[
	\eqref{equ:Gap_robust}\ \Rightarrow \ \eqref{equ:IACT_robust}\ \Rightarrow \ \eqref{equ:RJD_robust},
	\qquad
	\eqref{equ:Gap_robust}\ \Rightarrow \ \eqref{equ:AAR_robust}.	
\]
Hence, \eqref{equ:AAR_robust} and \eqref{equ:RJD_robust} are the most basic notions of concentration robust performance, i.e., if they fail to hold, then the other notions also fail to hold.
We, therefore, focus on studying whether \eqref{equ:AAR_robust} and \eqref{equ:RJD_robust} hold for suitable choices of proposal kernels $P_n$.

\subsection{Hessian-based Proposals}\label{sec:prop}
In this section, we propose two random walk-like proposal kernels $P_n$ for which we subsequently verify the properties \eqref{equ:AAR_robust} and \eqref{equ:RJD_robust} under suitable assumptions.
To this end, we construct proposal kernels which are informed about the concentration of the target measure.
Here, we adopt the humble idea of using an approximation $C_n \in \bbR^{d\times d}$ to the target covariance matrix for proposing new states.
This idea can be traced back at least to Tierney's seminal paper \cite{Tierney1994} and has been considered in many publications since.
For instance, the adaptive Metropolis algorithm by Haario et al. \cite{HaarioEtAl2001} approximates the target covariance empirically using the past states of the Markov chain and employs a scaled version of it in the proposal kernel.
Moreover, in \cite{MartinEtAl2012, Law2014, CuiEtAl2014, RudolfSprungk2018} the authors use gradient and Hessian information of the posterior density to construct approximations $C_n$ to the posterior covariance. 
We follow this approach and consider a simple offline (i.e., before running the Markov chain) computable proxy $C_n$ via the \emph{Laplace approximation} $\LAn$ of the target $\pi_n$:
\begin{equation}\label{equ:LA}
	\LAn \coloneqq \Nv\left(x_n, \frac 1n H^{-1}_n \right)
\end{equation}
where $x_n$ denotes the \emph{maximum a-posterior (MAP) estimate} given by
\begin{equation}\label{equ:MAP}
	x_n \coloneqq \argmin_{x\in \S_0} U(x) - \frac 1n \log \pi_0(x)
\end{equation}
with $\pi_0\colon \bbR^d \to [0,\infty)$ denoting the assumed Lebesgue density of $\pi_0$, and where
\begin{equation}\label{equ:Hn}
	H_n \coloneqq \nabla^2 U(x_n) - \frac 1n \nabla^2 \log \pi_0(x_n).
\end{equation}
The Gaussian measure $\LAn$ is derived from the classical Laplace's method for asymptotics of integrals \cite{Laplace1774} and is a commmon approximation in Bayesian inference \cite{GhoshEtAl2006} and Bayesian optimal experimental design \cite{alexanderian2016fast}.
Recently, convergence of $\LAn$ to $\pi_n$ for $n\to \infty$ has been shown in \cite{SchillingsEtAl2020} under suitable assumptions.
This convergence relates to the classical Bernstein--von Mises theorem and will be exploited in Section \ref{sec:NonGaussian_target}.

In order to have a well-defined $\LAn$ we make the following standing assumption.

\begin{assum}\label{assum:LA1}
The mappings $U, \pi_0\colon \S_0 \to [0,\infty)$ are twice continuously differentiable and for each $n\in\bbN$ there exists a unique minimizer $x_n$ as given in \eqref{equ:MAP} belonging to the interior of $\S_0$ with $H_n$ as in \eqref{equ:Hn} being positive definite.
\end{assum} 

Based on the proxy $\frac 1n H^{-1}_n$ to the covariance of the target $\pi_n$ we consider the following modifications of the Gaussian random walk and pCN proposal kernel:
\begin{enumerate}
\item
Hessian-based Gaussian random walk proposal:
\[
	P_n(x) = \Nv\left( x, \frac{s^2}n H^{-1}_n \right), \qquad s>0.
\]
This proposal kernel is reversible w.r.t.~the Lebesgue measure on $\bbR^d$ and, thus, leads to
\[
    \alpha_n(x,y) 
    =
    \alpha_{P_n}(x,y) 
    =
	\min\left\{1, \exp\left(n [U(x) - U(y)] \right)\, \frac{\pi_0(y)}{\pi_0(x)} \,\boldsymbol{1}_{\S_0\times \S_0}(x,y)\right\};
\]

\item
Modified pCN proposal, cf. \cite{PinskiEtAl2014, ChenEtAl2016}:
\[
	P_n(x) = \Nv\left( x_n + \sqrt{1-s^2} (x-x_n), \frac{s^2}{n} H^{-1}_n \right), \qquad s\in(0,1),
\]
which is reversible w.r.t.~the Laplace approximation $\LAn$ of $\pi_n$ and, thus, yields with $\varphi_{x_n,C_n}$ denoting the Lebesgue density of $\LAn = \Nv(x_n,C_n)$ where $C_n \coloneqq \frac 1n H_n^{-1}$
\begin{align*}
    \alpha_n(x,y) 
    & =
    \alpha_{P_n}(x,y) 
    = \min\left\{1, \exp\left(n[U(x) - U(y)] \right)\, \frac{\pi_0(y)}{\pi_0(x)}\, \frac{\varphi_{x_n,C_n}(x)}{\varphi_{x_n,C_n}(y)} \,\boldsymbol{1}_{\S_0\times \S_0}(x,y)  \right\}
\end{align*}
In case of a Gaussian prior $\pi_0 = \Nv(0,C)$, we can rewrite the acceptance probability as
\[
    \alpha_n(x,y) 
    =
	\min\left\{1, \exp\left(n [U(x) - U(y)] \right)\ \frac{\d \pi_0}{\d \LAn}(y) \ \frac{\d \LAn}{\d \pi_0}(x) \right\}
\]
and, thus, under suitable assumptions ensuring the equivalence of $\pi_0$ and  $\LAn$, this proposal yields again a well-defined MH algorithm in infinite-dimensional Hilbert spaces.
\end{enumerate}

Note, that the two proposals above possess the scaling $s_n = \frac sn$ of the stepsize with increasing concentration level, as derived in \cite{BeskosEtAl2018}.
However, due to $H_n^{-1}$ they also include geometric information on how fast the concentration occures in the various directions.
This information will prove to be essential in order to ensure \eqref{equ:RJD_robust}, and hence, allow for \eqref{equ:Gap_robust}. 
Although the scaling $s_n = \frac sn$ is sufficient to ensure \eqref{equ:AAR_robust}, we demonstrate for a simple example that just using such a scaling in a random walk proposal $P_n(x) = \Nv(x, s_n\, C)$ does not yield \eqref{equ:RJD_robust}.

\begin{exam}
Consider $\pi_n = \Nv\left(0, C_n\right)$ on $\bbR^2$ with $C_n = \diag(1, 1/(1+n))$, i.e.,
\[
    \pi_n(\d x)
    = \frac 1{Z_n} \exp(- n U(x)) \ \pi_0(\d x)
    \quad
    \text{with}
    \quad
    \pi_0 = \Nv(0, I_2),
    \
    U(x)
    =
    U(x_1,x_2)
    =
    \frac 12 x_2^2,
\]
then $P_n(x) = \Nv(x, s_n I)$ with $s_n = \frac sn$ for an $s>0$ yield to \eqref{equ:AAR_robust}, see Section \ref{sec:Gaussian_target}, but for $v = (1,0)^\top$ we have
\[
    \widebar{\rho}_v(K_n) 
    =
    \frac{\rho_v(K_n)}{\Var_{\pi_n}(f_v)}
    =
    \frac{s_n}{1}
    \xrightarrow[]{n\to\infty} 0.
\]
Thus, \eqref{equ:RJD_robust} and, hence, \eqref{equ:Gap_robust} do not hold.
\end{exam}

\section{Concentration Robustness of Hessian-based Gaussian Proposals}
\label{sec:VarRobust_Gauss}
In this section we present our main results stating that for the Hessian-based proposals introduced in Section \ref{sec:prop} the concentration robustness properties \eqref{equ:AAR_robust} and \eqref{equ:RJD_robust} hold under mild assumptions as well as \eqref{equ:Gap_robust} for the case of Gaussian targets.
Our strategy here consists of two steps: 1.) verifying the robustness results for the simple case of Gaussian targets by straightforward calculations and 2.) showing that these results extend to non-Gaussian targets by exploiting the convergence of concentrating posterior measures $\pi_n$ to their Laplace approximation $\LAn$ as $n\to\infty$.

\subsection{The case of Gaussian targets} \label{sec:Gaussian_target}
We consider target measures
\begin{equation}\label{equ:Gaussian_Target}
	\pi_n \coloneqq \Nv\left(x_n, C_n \right)
\end{equation}
with $x_n \in \bbR^d$ and $C_n \in \bbR^{n \times n}$ symmetric and positive definite.
Obviously, we have $\pi_n = \LAn$ in this case.
We prove that \eqref{equ:AAR_robust}, \eqref{equ:RJD_robust}, as well as \eqref{equ:Gap_robust} hold for MH algorithms targeting $\pi_n$ based on either the proposal $P_n(x) = \Nv\left( x, s^2 C_n\right)$, $s>0$, or $P_n(x) = \Nv\left( x_n + \sqrt{1-s^2} (x-x_n), s^2 C_n \right)$, $s\in(0,1)$.

To this end, we employ the methodology of \emph{pushforward transition kernels} which is discussed in detail in Appendix \ref{sec:pushforward}.
Let us summarize the most important results for our purposes from the appendix:

\begin{theo}[Appendix \ref{sec:pushforward}]\label{theo:app}
Let $\pi\in\mc P(\HS)$ and $T_*\pi \in \mc P(\HS)$ be the resulting pushforward measure for a measurable mapping $T\colon\HS\to\HS$.
Furthermore, assume a $\pi$-reversible transition kernel $K$ and define the corresponding \emph{pushfoward transition kernel $T_*K$} by
\[
	T_*K(x, A)
	\coloneqq
	\ev{K(X, T^{-1}(A))\mid T(X) = x},
	\qquad
	X\sim\pi,
\]
where $x\in\HS$, $A\in\mcH$.
Then,
\begin{enumerate}
    \item 
    $T_*K$ is $T_*\pi$-reversible
    
    \item
    we have $\gapm{\pi}{K} \leq  \gapm{T_*\pi}{T_*K}$,

    \item
    as well as $\widebar{\rho}_v(K) =  \widebar{\rho}_{T(v)}(T_*K)$ for any $v\in\HS$,
    
    \item
    if 
    $K$ is a MH transition kernel with proposal $P$, then the $T_*\pi$-reversible MH transition kernel with pushfoward proposal $T_*P$ coincides with $T_*K$ and $\widebar{\alpha}(K) = \widebar{\alpha}(T_*K)$.
\end{enumerate}
\end{theo}
We now show that the Hessian-based Gaussian random walk and the modified pCN proposal $P_n$ can be seen as pushforward proposal kernels and, thus, the corresponding $\pi_n$-reversible MH algorithms are pushfoward transition kernels of a reference MH transition kernel $K$ targeting $\pi = \Nv(0, I_d)$. 
This yields by means of Theorem \ref{theo:app} the resulting concentration robustness properties \eqref{equ:AAR_robust}, \eqref{equ:RJD_robust}, and \eqref{equ:Gap_robust} for Gaussian targets.

\paragraph{Hessian-based Gaussian Random Walk}
In the following we study the proposal
\begin{equation}\label{equ:GRW_Gauss}
	P_n(x) = \Nv\left( x, s^2 C_n\right), \qquad s>0.
\end{equation}
and the related $\pi_n$-reversible MH transition kernel $K_n$ where $\pi_n$ is as given in \eqref{equ:Gaussian_Target}.
A key observation for our analysis is the following.

\begin{propo}\label{propo:HRW_Gauss}
Let $\pi_n$ be as in \eqref{equ:Gaussian_Target} and consider the proposal kernel \eqref{equ:GRW_Gauss}.
Then, we have for the resulting $\pi_n$-reversible MH kernel $K_n$ that
\[
	K_n = (T_n)_* K,
	\qquad
	T_n(x) \coloneqq x_n + C_n^{1/2}x,
\]
where $K$ denotes the $\pi$-reversible MH transition kernel targeting $\pi = N(0,I_d)$, using the proposal kernel $P(x) = \Nv(x, s^2 I_d)$.
\end{propo}
\begin{proof}
Obviously, $\pi_n = (T_n)_*\pi$.
Thus, it is sufficient by Theorem \ref{theo:app} to show that $P_n = \left(T_n\right)_* P$.
To this end, let $x\in\HS$ and $B\in \mcH$.
We obtain with $\xi \sim N(0,I_d)$
\begin{align*}
	P_n(x,B)
	& = \bbP\left( x + s C_n^{1/2}\xi \in B\right)
	= \bbP\left( s\xi \in C_n^{-1/2}(B - x)\right)\\
	& = \bbP\left( s\xi \in C_n^{-1/2}(B - x_n) - C_n^{-1/2}(x-x_n)\right)\\
	& = \bbP\left(  C_n^{-1/2}(x-x_n) + s\xi \in C_n^{-1/2}(B - x_n)\right)\\
	& = P\left(T_n^{-1}(x),T_n^{-1}(B)\right),
\end{align*}
which yields the assertion.
\end{proof}

\paragraph{Modified pCN}
Next, we study the modified pCN proposal
\begin{equation}\label{equ:pCN_Gauss}
	P_n(x) = \Nv\left( x_n + \sqrt{1-s^2}(x - x_n), s^2 C_n\right), \qquad s\in(0,1],
\end{equation}
and the related $\pi_n$-reversible MH transition kernel $K_n$ where $\pi_n$ is as in \eqref{equ:Gaussian_Target}.
Similarly to Proposition \ref{propo:HRW_Gauss} we have the following basic result.

\begin{propo}\label{propo:pCN_Gauss}
Let $\pi_n$ be as in \eqref{equ:Gaussian_Target} and consider the proposal kernel \eqref{equ:pCN_Gauss}.
Then, we have for the resulting $\pi_n$-reversible MH kernel $K_n$ that
\[
	K_n = (T_n)_* K,
	\qquad
	T_n(x) \coloneqq x_n + C_n^{1/2}x,
\]
where $K$ denotes the $\pi$-reversible MH transition kernel targeting $\pi = N(0,I_d)$ and using the pCN proposal kernel $P(x) = \Nv(\sqrt{1-s^2}x, s^2 I_d)$.
\end{propo}
The proof is analogous to the one of Proposition \ref{propo:HRW_Gauss} and therefore omitted.
As a consequence of the above two propositions and Theorem \ref{theo:app}, we obtain the following.

\begin{theo}[Concentration robustness for Gaussian target]\label{theo:HRW_Gauss}
Let $\pi_n$ be as in \eqref{equ:Gaussian_Target}.
Then, 
\begin{enumerate}
    \item 
    for the $\pi_n$-reversible MH transition kernel $K_n$ using the Hessian-based Gaussian random walk proposal kernel \eqref{equ:GRW_Gauss} we have
    \[
	\widebar \alpha(K_n)
	=
	\ev{1 \wedge \exp\left(- \frac 12 \|X + s\xi\|^2 + - \frac 12 \|X\|^2\right)} > 0
    \]
    with $X,\xi \sim N(0,I_d)$ independently and $a\wedge b \coloneqq\min\{a,b\}$ for $a,b\in\bbR$, as well as
    \[
	\widebar \rho_{v}(K_n)
	=
	\ev{s^2 \xi_1^2 \left(1 \wedge \exp\left(- \frac 12 \|X + s\xi\|^2 + - \frac 12 \|X\|^2\right)\right) } >0
    \]
    for any $0\neq v\in\bbR^d$ where $\xi_1$ denotes the first random component of $\xi$ as above;

    \item
    for the $\pi_n$-reversible MH transition kernel $K_n$ using the modified pCN proposal kernel \eqref{equ:pCN_Gauss} we have
    \[
	\widebar \alpha(K_n)
	=
	1,
	\qquad
	\widebar \rho_{v}(K_n)
	=
	 2- 2\sqrt{1-s^2}, 
    \]
    for any $0\neq v\in\bbR^d$;
    
    \item
    for both aforementioned $\pi_n$-reversible MH transition kernels $K_n$ we have
    \[
    	\gapm{\pi_n}{K_n}
    	=
    	\text{const.} > 0,
    \]
    i.e., the spectral gap of $K_n$ is independent of $n$.
\end{enumerate}
\end{theo}
\begin{proof}
The third statement follows immediately by Theorem \ref{theo:app}.
Regarding the first statement we have by Proposition \ref{propo:HRW_Gauss} $K_n = (T_n)_*K$ where the acceptance probability of $K$ is given by
\[
	\alpha(x,y) = 1 \wedge \exp\left(- \frac 12 \|y\|^2 + \frac 12 \|x\|^2\right).
\]
Thus, 
\[
	\widebar \alpha(K)
	=
	\ev{1 \wedge \exp\left(- \frac 12 \|X + s\xi\|^2 + \frac 12 \|X\|^2\right)}
\]
with $X,\xi \sim N(0,I_d)$ independently.
Moreover, for $f_v(x) \coloneqq v^\top x$ we get $f_v\circ T_n(x) = v^\top x_n + (C_n^{1/2}v)^\top x$.
Now, Theorem \ref{theo:app} yields $\widebar \rho_v(K_n) = \widebar \rho_{C_n^{1/2}v}(K)$. 
Due to $\Var_{\pi}(f_v) = v^\top I_d v = \|v\|^2$ for $\pi = N(0,I_d)$, we have $\widebar \rho_{v}(K) = \widebar \rho_{v/\|v\|}(K)$.
Moreover, due to the radial symmetry of the measure $K(x,\d y) \ \pi(\d x)$ we obtain $\widebar \rho_{v/\|v\|}(K) = \widebar\rho_{e_1}(K)$ with $e_1$ denoting the first canonical unit vector in $\bbR^d$.
Thus, by Theorem \ref{theo:app}
\[
	\widebar \rho_v(K_n)
	=
	\widebar\rho_{e_1}(K)
	=
	\ev{s^2 \xi_1^2 \left(1 \wedge \exp\left(- \frac 12 \|X + s\xi\|^2 + - \frac 12 \|X\|^2\right)\right) }
\]
where $\xi = (\xi_1,\ldots,\xi_d)$. 

Regarding the MH transition kernel $K_n$ the statement $\widebar \alpha(K_n) = 1$ follows by the $\pi_n$-reversibility of the proposal kernel.
Exploiting the same arguments as above, we further get
\[
	\widebar \rho_v(K_n) 
	= \widebar \rho_v(P_n) 
	= \widebar \rho_{C_n^{1/2}v}(P)
	= \widebar \rho_{e_1}(P).
\]
Hence, with $X,\xi \sim N(0,1)$ we obtain
\[
	\widebar \rho_{e_1}(P) 
	= \ev{ \left(X - \left(\sqrt{1-s^2} X + s \xi\right)\right)^2 }
	= 2- 2\sqrt{1-s^2},
\]
which concludes the proof
\end{proof}


\subsection{Extension to non-Gaussian posteriors}\label{sec:NonGaussian_target}
In this section we lift the results on concentration robustness established in the previous section for the Hessian-based proposals introduced in Section \ref{sec:VarRobust_Intro} to non-Gaussian target measures $\pi_n$.
A key fact to do so is the convergence of the Laplace approximation $\LAn$ to $\pi_n$ as $n\to\infty$ in Hellinger distance presented in \cite{SchillingsEtAl2020}.
Thus, even non-Gaussian targets $\pi_n$ tend to be Gaussian as $n\to\infty$, and, thus, the results for Gaussian targets can be extended to that case.
We summarize the key points of our strategy as follows:
\begin{enumerate}
    \item
    The MH transition kernels $\widetilde K_n$ targeting the Laplace approximation $\LAn$ of $\pi_n$ and using the Hessian-based proposals $P_n$ of Section \ref{sec:VarRobust_Intro}, satisfy the assumptions of Theorem \ref{theo:HRW_Gauss} and, therefore, have a concentration robust performance in terms of, e.g., $\widebar \alpha(\widetilde K_n)$ and $\widebar \rho_{v}(\widetilde K_n)$.
    
    \item
    For the MH transition kernels $K_n$ targeting the concentrating posteriors $\pi_n$ using the same Hessian-based proposals $P_n$ we show that the difference of the corresponding efficiency quantities $|\widebar \alpha(K_n)-\widebar \alpha(\widetilde K_n)|$ and $|\widebar \rho_{v}(K_n) - \widebar \rho_{v}(\widetilde K_n)|$ can be bounded by the total variation and Hellinger distance, respectively, of the different targets $\pi_n$ and $\LAn$.
    
    \item
    Thus, given that $\LAn$ and $\pi_n$ converge to each other in the suitable sense, the efficiency quantities $\widebar \alpha$ and $\widebar \rho_v$ of $K_n$ and $\widetilde K_n$, respectively, also converge to each other and, thus, also $K_n$ performs concentration robustly in terms of $\widebar \alpha(K_n)$ and $\widebar \rho_{v}(K_n)$.
\end{enumerate}

To this end, we recall some recent results on the convergence of the Laplace approximation.

\subsubsection{Convergence of the Laplace approximation}
We start with classical results on the weak convergence of the Laplace approximation based on Laplace's method for integrals.

\begin{theo}[variant of {\cite[Section IX.5]{Wong2001}}]
\label{theo:laplace_method}
Set
\[
	J(n)
	\coloneqq
	\int_D f(x) \exp(-n U(x))\ \d x,
	\qquad
	n\in\bbN,
\]
where $D\subseteq \bbR^d$ is a possibly unbounded domain and where the integral $J(n)$ converges absolutely for each $n\in\bbN$.
Then, given that there exists an $x_\star$ in the interior of $D$ such that for every $r > 0$ we have
\[
	0 < \inf_{x\colon \|x-x_\star\|>r} U(x) - U(x_\star),
\]
and that in a neighborhood of $x_\star$ the function $f:D\to \mathbb R$ is twice continuously differentiable, $U\colon \bbR^d \to \bbR$ is three times continuously differentiable, and the Hessian $H_\star \coloneqq \nabla^2 U(x_\star)$ is positive definite, we have as $n\to \infty$
\[
	J(n) 
	=
	\e^{-n U(x_\star)}\,
	n^{-d/2}\,\left(
	\sqrt{\det(2\pi H^{-1}_\star)}\, f(x_\star)
	+ \mc O\left(n^{-1}\right)\right).
\]
\end{theo}
In fact, if $f:D\to \mathbb R$ is $(2p+2)$ times continuously differentiable and $U\colon \bbR^d \to \bbR$ is $(2p+3)$ times continuously differentiable, then $J(n)$ allows for an asymptotic expansion 
\[
	J(n) 
	=
	\e^{-n U(x_\star)}\,
	n^{-d/2}\,
	\left(\sum_{k=0}^p c_k(f) n^{-k}
	+ \mc O\left(n^{-p-1}\right) \right)
\]
with $c_0(f) = \sqrt{\det(2\pi H^{-1}_\star)}\, f(x_\star)$ and explicit formulas for $c_k(f)$, $k\geq 1$, see \cite{SchillingsEtAl2020,Wong2001}.
Based on these expansions, one can show the following asymptotic behaviour of $\pi_n(f)$ and $\Var_{\pi_n}(f)$.

\begin{propo}[{\cite[Section 3]{SchillingsEtAl2020}}]\label{propo:Schillings}
Let $\pi_n$ be given as in \eqref{equ:post} where $\pi_0$ allows for a Lebesgue density denoted also by $\pi_0\colon \bbR^d \to [0,\infty)$.
Given the assumptions of Theorem \ref{theo:laplace_method} for an $f \in L^2_{\pi_0}(\bbR)$ and, moreover, that $x_\star$ belongs to the interior of $\S_0 = \{x \in \bbR^d\colon \pi_0(x) >0\}$ as well as $\pi_0$ being twice continuously differentiable in a neighborhood of $x_\star$, we have 
\begin{align} 
    \label{equ:Exp_mu}
	\pi_n(f)
	& = f(x_\star) + \mc O(n^{-1}),
	\qquad
	\Var_{\pi_n}(f)
	\in
	\mc O(n^{-1}).
\end{align}
Furthermore, if $f$ and $\pi_0$ are four times continuously differentiable and $U$ five times continuously differentiable in a neighborhood of $x_\star$, then \begin{align} \label{equ:Var_mu_2}
	\Var_{\pi_n}(f)
	=
	n^{-1} \|\nabla f(x_\star)\|^2_{H_\star^{-1}} + \mc O(n^{-2}).
\end{align}
\end{propo}
These results yield kind of a weak convergence of $\pi_n$ to $\LAn$ (both converging weakly to $\delta_{x_\star}$) for sufficiently smooth functionals $f$. 
However, we require for the following a stronger convergence of $\pi_n$ to $\LAn$, namely, in Hellinger distance.
This was recently shown in \cite{SchillingsEtAl2020}.
For convenience, we recall the definition of the total variation (TV) and Hellinger distance:
\begin{align*}
	d_\text{TV}(\pi,\widetilde \pi)
	& \coloneqq 
	\sup_{A \in \mcH}
	\left|\pi(A) - \widetilde\pi(A)\right|
	=
	\frac 12 \int_{\bbR^d} \left|\frac{\d \pi}{\d \mu}(x) - \frac{\d \widetilde \pi}{\d \mu}(x)\right| \mu(\d x)\\
	d_\text{H}(\pi,\widetilde \pi),
	&
	\coloneqq 
	\left(\int_{\bbR^d} \left|\sqrt{\frac{\d \pi}{\d \mu}(x)} - \sqrt{\frac{\d \widetilde \pi}{\d \mu}(x)}\right|^2 \mu(\d x) \right)^{1/2},
\end{align*}
where $\mu$ denotes a common dominating measure for $\pi,\widetilde \pi \in \mc P(\HS)$.
It holds true that
\[
	\frac{d^2_\text{H}(\pi,\widetilde \pi)}2
	\leq
	d_\text{TV}(\pi,\widetilde \pi)
	\leq
	d_\mathrm{H}(\pi,\widetilde \pi),
\]
see, e.g., \cite[Equation (8)]{GibbsSu2002}. 
Moreover, 
\begin{align}\label{eq:Hell_Mom}
    \left|\pi(f) - \widetilde \pi(f) \right|
    & \leq
    2 \left(\pi(|f|^2) + \widetilde \pi(|f|^2) \right)^{1/2} \ d_\text{H}(\pi, \widetilde \pi),
    \qquad
    f \in L^2_\pi \cap L^2_{\widetilde \pi},
\end{align}
see, e.g., \cite[Lemma 21]{DashtiStuart2017}.

\begin{theo}[{\cite[Theorem 2]{SchillingsEtAl2020}}]\label{theo:conv_H}
Let $\pi_n$ be given as in \eqref{equ:post} with $\pi_0$ possessing a Lebesgue density $\pi_0\colon \bbR^d \to [0,\infty)$. 
Furthermore, let $U, \pi_0 \in C^3(\S_0,\bbR)$ with $\S_0 = \{x \in \bbR^d\colon \pi_0(x) >0\}$ and let Assumption \ref{assum:LA1} be satisfied.
If
\begin{enumerate}
\item
there exist the limit $x_\star \coloneqq \lim_{n\to \infty} x_n$ within the interior of $\S_0$ and $H_\star \coloneqq \nabla^2U(x_\star)$ being positive definite;

\item
for each $r > 0$ there exists an $n_r\in\bbN$ such that
\[
	\inf_{x \in B^c_r(x_n) \cap \S_0} U(x) - \frac 1n \log \pi_0(x) > 0  \qquad \forall n\geq n_r,
\]
with $B^c_r(x_n) \coloneqq \{x \in \HS\colon \|x-x_n\| > r\}$;

\item
the prior density satisfies $\int_{\HS} \pi_0^{1-\epsilon}(x)\ \d x < \infty$ for an $\epsilon\in(0,1)$;

\end{enumerate}
then, there holds
\[
		d_\mathrm{H}(\mu_n, \LAn) \in \mc O(n^{-1/2}).
\]
\end{theo}

The first condition of Theorem \ref{theo:conv_H} relates the (pathwise) convergence of the MAP estimate $x_n$ to the maximum-Likelihood estimate $x_\star$ in the small noise limit.
The second assumptions ensures that the maximum of the density of $\pi_n$ at $x_n$ is well separated from any other local peaks and the third assumptions is rather mild and a technical requirement for the proof of the theorem.

A rather strong assumption for practical applications here is the positive definiteness of $H_\star \coloneqq \nabla^2U(x_\star)$ which yields convergence of $\pi_n$ to a point mass located at $x_\star$.
This assumption can be relaxed under suitable conditions, namely, if the prior $\pi_0$ is Gaussian and if $\pi_n$ concentrates around a linear manifold.

\begin{corollary}[{\cite[Corollary 1]{SchillingsEtAl2020}}]\label{cor:conv_H_sing}
Let $\pi_n$ be given as in \eqref{equ:post} with $\pi_0 = \Nv(x_0,C_0)$ and let Assumption \ref{assum:LA1} be satisfied.
If there exists a linear subspace $\mc X \subset \HS$ such that for the orthogonal projection $\mathrm P_{\mc X} \colon \HS \to \mc X$ we have
\[
	 U = U \circ \mathrm P_{\mc X} \qquad \text{ on } \HS,
\]
and if the restriction $U\colon \mc X \to [0,\infty)$ of $U$ on $\mc X$ and marginal prior density $\pi_0$ on $\mc X$ satisfy the assumptions of Theorem \ref{theo:conv_H} on $\mc X$, then we have $d_\mathrm{H}(\pi_n, \LAn) \in \mc O(n^{-1/2})$.
\end{corollary}
Let $\mc X^\perp$ denote the orthogonal complement to the linear subspace $\mc X$.
Then, the assumptions of Corollary \ref{cor:conv_H_sing} yield that the posterior $\pi_n$ concentrates along the linear manifold $\mc M = x_\star + \mc X^\perp = \argmin_{x\in\HS} U(x)$.
Similarly, the Laplace approximation concentrates along $\mc M$. In particular, given the assumption of Corollary \ref{cor:conv_H_sing} we have for the covariance $C_n = \frac 1n \left( \nabla^2 U(x_n) + \frac 1n C_0^{-1} \right)^{-1}$ of $\LAn$ that
\begin{align} \label{equ:Var_LAn_sing}
    \lim_{n\to \infty} n\ v^\top C_n v
    < \infty
    \quad
    \text{iff}
    \quad
    v \in \mc X,
\end{align}
i.e., the marginal variance of $\LAn$ decays like $\mc O(n^{-1})$ just in the directions of the `likelihood-informed' subspace $\mc X$. 
The statement of \eqref{equ:Var_LAn_sing} is based on the fact that the null space of $\nabla^2 U(x_n) \in \bbR^{d\times d}$ coincides for sufficiently large $n\in\bbN$ with $\mc X^\perp$ as defined above, and can be verified rigorously by a straightforward computation.

\subsubsection{Main result}
Given the results for $d_\mathrm{H}(\pi_n, \LAn)\to0$ we are now ready to state our main result for the concentration robustness of MH algorithms using Hessian-based random walk-like proposals.

\begin{theo}[Concentration robustness for non-Gaussian targets]\label{theo:HRW_NonGauss}
Let $\pi_n$ be as in \eqref{equ:post} and let $\pi_0 \in C^4(\S_0;\bbR)$ and $U \in C^5(\S_0;\bbR)$. 
Further, let either the assumptions of Theorem \ref{theo:conv_H} be satisfied with 
\[
	0 < \inf_{x\colon \|x-x_\star\|>r} U(x) - U(x_\star) \qquad \forall r >0
\]
or the assumptions of Corollary \ref{cor:conv_H_sing} such that $x_\star \coloneqq \lim_{n\to\infty} \PM_{\mc X} x_n$ satisfies
\[
	0 < \inf_{x \in \mc X \colon \|x - x_\star\| > r} U(x) - U(x_\star) \qquad \forall r >0.
\]
Then, 
\begin{enumerate}
    \item 
    for the $\pi_n$-reversible MH transition kernel $K_n$ using the Hessian-based Gaussian random walk proposal kernel \eqref{equ:GRW_Gauss} we have
    \[
	\lim_{n\to \infty} \widebar \alpha(K_n)
	=
	\ev{1 \wedge \exp\left(- \frac 12 \|X + s\xi\|^2 + - \frac 12 \|X\|^2\right)} > 0
    \]
    with $X,\xi \sim N(0,I_d)$ independently and $a\wedge b \coloneqq\min\{a,b\}$ for $a,b\in\bbR$, as well as
    \[
	\lim_{n\to \infty} \widebar \rho_{v}(K_n)
	=
	\ev{s^2 \xi_1^2 \left(1 \wedge \exp\left(- \frac 12 \|X + s\xi\|^2 + - \frac 12 \|X\|^2\right)\right) } >0
    \]
    for any $0\neq v\in\bbR^d$ where $\xi_1$ denotes the first random component of $\xi$ as above;

    \item
    for the $\pi_n$-reversible MH transition kernel $K_n$ using the modified pCN proposal kernel \eqref{equ:pCN_Gauss} we have
    \[
	\lim_{n\to \infty} \widebar \alpha(K_n)
	=
	1,
	\qquad
	\lim_{n\to \infty} \widebar \rho_{v}(K_n)
	=
	 2- 2\sqrt{1-s^2}, 
    \]
    for any $0\neq v\in\bbR^d$.
\end{enumerate}
\end{theo}
\begin{proof}
Let $P_n$ be either of the two proposal kernels and $K_n$ denote the resulting MH transition kernel targeting $\pi_n$ as well as $\widetilde K_n$ the resulting MH transition kernel targeting $\LAn$.
By Lemma \ref{lem:AAR_pi_tilde_pi} shown in the subsequent section, we then have 
\[
    \left|\widebar \alpha(K_n) - \widebar \alpha(\widetilde K_n)\right| 
    \leq 
    2d_\text{TV}(\pi_n, \LAn)
    \leq
    2d_\text{H}(\pi_n, \LAn)
    \to 0
\]
as $n\to \infty$ exploiting Theorem \ref{theo:conv_H} or Corollary \ref{cor:conv_H_sing}, respectively.
Thus, we have
\[
    \lim_{n\to \infty} \widebar \alpha(K_n)
    =
    \lim_{n\to \infty} \widebar \alpha(\widetilde K_n)
\]
and the statement follows by Theorem \ref{theo:HRW_Gauss} applied to $\widetilde K_n$.

Regarding the directional expected squared jump distances $\widebar{\rho}_v(K_n)$ and $\widebar{\rho}_v(\widetilde K_n)$, respectively, we have by Lemma \ref{lem:SJD_pi_tilde_pi} and Proposition \ref{propo:SJD_Proposal} shown in the subsequent section that
\[
    \left|\widebar \rho_v(K_n) - \widebar \rho_v(\widetilde K_n)\right|
    \leq
    C
    \ 
    d_\text{H}(\pi_n,\LAn)
    +
	\widebar \rho_f(K) \ 
	\left|1 - \frac{\Var_{\pi_n}(f_v)}{\Var_{\LAn}(f_v)} \right|
\]
where $f_v(x) = v^\top x$ and $C<\infty$. 
By Proposition \ref{propo:Var_pin LAn} we know that $\lim_{n\to\infty} \frac{\Var_{\pi_n}(f_v)}{\Var_{\LAn}(f_v)} = 1$, and, hence, by the same reasoning as above we obtain 
\[
    \lim_{n\to \infty} \widebar{\rho}_v(K_n)
    =
    \lim_{n\to \infty} \widebar{\rho}_v(\widetilde K_n)
\]
and the statement follows again by Theorem \ref{theo:HRW_Gauss}.
\end{proof}

In the following subsection we collect all auxiliary results required for the proof of our main Theorem~\ref{theo:HRW_NonGauss}.

\subsubsection{Auxiliary Results}\label{sec:AAR}
\paragraph{Stability of the average acceptance rate.}
We first provide a general result regarding the average acceptance rate of MH transition kernels targeting two different probability measures $\pi, \widetilde \pi$ on $\bbR^d$ using the same proposal kernel $P\colon \bbR^d \times \mc B^d\to[0,1]$.
We introduce the following measures on $\bbR^d\times\bbR^d$:
\begin{align}\label{equ:nu}
	\nu(\d x\, \d y)
	\coloneqq
	P(x, \d y)\ \pi(\d x),
	\qquad
	\nu^\top(\d x\, \d y)
	\coloneqq
	\nu(\d y\, \d x)
\end{align}
as well as
\begin{align}\label{equ:nu_tilde}
	\widetilde \nu(\d x \d y)
	\coloneqq
	P(x, \d y)\ \widetilde \pi(\d x),
	\qquad
	\widetilde \nu^\top(\d x \d y)
	\coloneqq
	\widetilde  \nu(\d y \d x).
\end{align}
Thus, the acceptance probability $\alpha$ of the $\pi$-reversible MH transition kernel $K$ and the acceptance probability $\widetilde \alpha$ of the $\widetilde \pi$-reversible MH transition kernel $\widetilde K$, both employing the proposal kernel $P$, are given by 
\begin{align} \label{eq:AR_K_tilde_K}
	\alpha(x,y)
	& \coloneqq
	1 \wedge \frac{\d \nu^\top}{\d \nu}(x,y),
	&\widetilde
	\alpha(x,y)
	& \coloneqq
	1 \wedge \frac{\d \widetilde\nu^\top}{\d \widetilde\nu}(x,y)
\end{align}
assuming that $\nu^\top \ll \nu$ and $\widetilde \nu^\top \ll \widetilde \nu$ and recalling the notation $a\wedge b \coloneqq\min\{a,b\}$ for $a,b\in\bbR$.
For the corresponding average acceptance rates
\begin{align} \label{eq:AAR_K_tilde_K}
	\widebar \alpha(K)
	& =
	\int_{\HS} \int_{\HS}
	\alpha(x,y)\ P(x,\d y)\ \pi(\d x),
	&
	\widebar\alpha(\widetilde K)
	& =
	\int_\HS \int_\HS
	\widetilde \alpha(x,y)\ P(x,\d y)\ \widetilde \pi(\d x)
\end{align}
we obtain the following stability result.

\begin{lem}\label{lem:AAR_pi_tilde_pi}
Given $\alpha, \widetilde \alpha$ in \eqref{eq:AR_K_tilde_K} are well-defined, we have for $\widebar \alpha(K)$ and $\widebar\alpha(\widetilde K)$ as given in \eqref{eq:AAR_K_tilde_K}
\[
    \left|\widebar \alpha(K) - \widebar \alpha(\widetilde K)\right| 
    \leq 
    2d_\text{TV}(\pi, \widetilde \pi).
\]
\end{lem}
\begin{proof}
First, we note that
\[
	\widebar \alpha(K) = \nu \wedge \nu^\top(\HS \times \HS),
	\qquad
	\widebar\alpha(\widetilde K) = \widetilde \nu \wedge \widetilde \nu^\top(\HS \times \HS),
\]
where for two measures $\eta_1, \eta_2 \gg \eta$ on $\bbR^d$ with densities $h_i \coloneqq \frac{\d \eta_i}{\d \eta}$ we define the measure $\eta_1 \wedge \eta_2$ on $\bbR^d$ by
\[
	\eta_1 \wedge \eta_2(A) 
	\coloneqq
	\int_A \left(h_1 \wedge h_2\right) \ \d\eta,
	\qquad
	A\in\mc B^d,
\]
which yields $\eta_1 \wedge \eta_2(A) \leq \eta_1(A) \wedge \eta_2(A)$.
Note, that
\begin{align}\label{eq:dTV_nu_pi}
	d_\text{TV}(\nu^\top, \widetilde \nu^\top)
	=	
	d_\text{TV}(\nu, \widetilde \nu)
	=
	d_\text{TV}(\pi, \widetilde \pi),
\end{align}
where the first equality follows by construction and the second by considering a dominating measure $\mu \gg \pi, \widetilde\pi$ and computing the TV distance $d_\text{TV}(\nu, \widetilde \nu)$ with respect to the reference measure $\eta(\d x\, \d y) \coloneqq P(x,\d y) \ \mu(\d x)$, since $\eta \gg \nu, \widetilde \nu$ with $\frac{\d \nu}{\d \eta} = \frac{\d \pi}{\d\mu}$ and $\frac{\d \widetilde \nu}{\d \eta} = \frac{\d \widetilde \pi}{\d\mu}$.
Now, let $\eta$ denote an arbitrary dominating measure $\eta \gg \nu, \nu^\top, \widetilde \nu, \widetilde \nu^\top$ and consider the densities $h \coloneqq \frac{\d \nu}{\d \eta}$, $h^\top \coloneqq \frac{\d \nu^\top}{\d \eta}$, $\widetilde h \coloneqq \frac{\d \widetilde\nu}{\d \eta}$, and $\widetilde h^\top \coloneqq \frac{\d \widetilde \nu^\top}{\d \eta}$.
Since for $a_1,a_2,b_1,b_2 \in \bbR$
\[
	\left| a_1 \wedge a_2 - b_1 \wedge b_2 \right|
	\leq
	|a_1 - b_1| \vee |a_2 - b_2|
	\leq
	|a_1 - b_1| + |a_2 - b_2|
\]
we obtain
\begin{align*}
	d_\text{TV}(\nu \wedge \nu^\top, \widetilde\nu \wedge \widetilde \nu^\top)
	& =
	\frac 12 \int_{\HS\times \HS} 
	\left| h \wedge h^\top - \widetilde h \wedge \widetilde h^\top \right| \ \d \eta
	\leq \frac 12  \int_{\HS\times \HS} \left| h - \widetilde h \right| + \left| h^\top - \widetilde h^\top \right| \d \eta\\
	& = d_\text{TV}(\nu , \widetilde\nu ) + d_\text{TV}(\nu^\top , \widetilde\nu^\top)\\
	& \leq 2 d_\text{TV}(\pi, \widetilde \pi)
\end{align*}
and, thus, the statement follows by 
\[
    \left|\widebar \alpha(K) - \widebar \alpha(\widetilde K)\right|
    =
    \left|\nu \wedge \nu^\top(\HS \times \HS) - \widetilde \nu \wedge \widetilde \nu^\top(\HS \times \HS)\right|
    \leq
    d_\text{TV}(\nu \wedge \nu^\top, \widetilde\nu \wedge \widetilde \nu^\top).
\]
\end{proof}


\paragraph{Stability of the expected squared jump distance.}
Again we consider two arbitrary target probability measures $\pi, \widetilde \pi$ on $\bbR^d$ and an arbitrary proposal kernel $P\colon \bbR^d \times \mc B^d\to[0,1]$ such that the corresponding acceptance probabilities $\alpha$ and $\widetilde \alpha$ as given in \eqref{eq:AAR_K_tilde_K} are well-defined.
Let $K$ and $\tilde K$ denote the corresponding MH transition kernels targeting $\pi$ and $\widetilde \pi$, respectively, with proposal kernel $P$ and acceptance probability $\alpha$ and $\widetilde \alpha$, respectively.
We then consider a measurable function $f\colon \HS \to \bbR$ and the associated \emph{normalized $f$-jump distance} for $K$ and $\widetilde K$, respectively, given by
\begin{align}\label{eq:FJP_K_tilde_K}
	\widebar \rho_f(K) 
	& \coloneqq
	\frac{\int_{\HS\times \HS} \left| f(x) - f(y)\right|^2 \ \alpha(x,y)\ P(x, \d y) \ \pi(\d x)}{\Var_{\pi}(f)},\\
    \widebar \rho_f(\widetilde K) 
	& \coloneqq
	\frac{\int_{\HS\times \HS} \left| f(x) - f(y)\right|^2 \ \widetilde \alpha(x,y)\ P(x, \d y) \ \widetilde \pi(\d x)}{\Var_{\widetilde \pi}(f)}.
\end{align}

\begin{lem}\label{lem:SJD_pi_tilde_pi}
Let $f\colon \bbR^d\to\bbR$ belong to $L^2_{\pi}(\bbR) \cap L^2_{\widetilde \pi}(\bbR)$ and set $\Delta f(x,y) \coloneqq f(x) - f(y)$. 
Then we have
\[
    \left|\widebar \rho_f(K) - \widebar \rho_f(\widetilde K)\right|
    \leq
    \frac{4\left( \nu (|\Delta f|^4) + \widetilde \nu (|\Delta f|^4)\right)^{1/2}}{\Var_{\widetilde\pi}(f)}
    \ 
    d_\text{H}(\pi,\widetilde\pi)
    +
	\widebar \rho_f(K) \ 
	\left|1 - \frac{\Var_{\pi}(f)}{\Var_{\widetilde\pi}(f)} \right|
\]
with $\nu$ and $\widetilde \nu$ as in \eqref{equ:nu} and \eqref{equ:nu_tilde}, respectively.
\end{lem}
\begin{proof}
First recall from above that $[\nu \wedge \nu^\top](\d x\, \d y) \coloneqq \alpha(x,y)\ P(x,\d y)\ \pi(\d x)$, analogously for $\widetilde \nu \wedge \widetilde \nu^\top$.
By the definition of $\Delta f$
we have
\[
	\widebar \rho_f(K) 
	=
	\frac{\int_{\HS^2} |\Delta f|^2 \ \d [\nu \wedge \nu^\top]}{\Var_{\pi}(f)},
	\qquad
	\widebar \rho_f(\widetilde K) 
	=
	\frac{\int_{\HS^2} |\Delta f|^2 \ \d [\widetilde \nu \wedge \widetilde\nu^\top]}{\Var_{\widetilde\pi}(f)}.
\]
Since for $a_1,a_2,b_1,b_2 \in \bbR$ with $b_1,b_2 \neq 0$
\[
	\left| \frac{a_1}{b_1} - \frac{a_2}{b_2} \right|
	\leq
	\left| \frac{a_1}{b_1} - \frac{a_2}{b_1} \right|
	+
	\left| \frac{a_2}{b_1} - \frac{a_2}{b_2} \right|
	=
	\frac{|a_1-a_2|}{|b_1|}
	+
	\frac{|a_2|}{|b_2|}
	\left| 1 - \frac{b_2}{b_1} \right|
\]
we obtain
\[
	\left|
	\widebar \rho_f(K)
	-
	\widebar \rho_f(\widetilde K)
	\right|
	\leq
	\frac{\left|\int_{\HS^2} |\Delta f|^2 \ \d [\nu \wedge \nu^\top] - \int_{\HS^2} |\Delta f|^2 \ \d [\widetilde \nu \wedge \widetilde \nu^\top]\right| }{\Var_{\tilde \pi}(f)}
	+
	\widebar \rho_f(K)
	\left|1 - \frac{\Var_{\pi}(f)}{\Var_{\widetilde\pi}(f)} \right|.
\]
The assertion follows now by applying \eqref{eq:Hell_Mom} which also holds for general nonnegative measures such as $\nu \wedge \nu^\top$, $\widetilde \nu \wedge \widetilde \nu^\top$:
\begin{align*}
	\left| \nu \wedge \nu^\top(|\Delta f|^2) - \widetilde \nu \wedge \widetilde \nu^\top(|\Delta f|^2) \right|
	\leq
	2\left( \nu \wedge \nu^\top(|\Delta f|^4) + \widetilde \nu \wedge \widetilde \nu^\top(|\Delta f|^4)\right)^{1/2}\
	d_\text{H}(\nu \wedge \nu^\top, \widetilde\nu \wedge \widetilde \nu^\top).
\end{align*}
Note, that by the same arguments for deriving \eqref{eq:dTV_nu_pi} we have
\begin{align}\label{eq:dH_nu_pi}
	d_\text{H}(\nu^\top, \widetilde \nu^\top)
	=	
	d_\text{H}(\nu, \widetilde \nu)
	=
	d_\text{H}(\pi, \widetilde \pi)
\end{align}
and by
\[
	\left| \sqrt{a_1 \wedge a_2} - \sqrt{b_1 \wedge b_2} \right|
	\leq
	|\sqrt{a_1} - \sqrt{b_1}| \vee |\sqrt{a_2} - \sqrt{b_2}|
	\leq
	|\sqrt{a_1} - \sqrt{b_1}| + |\sqrt{a_2} - \sqrt{b_2}|
\]
for arbitrary $a_1,a_2,b_1,b_2 \geq 0$, we obtain analogously to the proof of Lemma \ref{lem:AAR_pi_tilde_pi}
\[
    d_\text{H}(\nu \wedge \nu^\top, \widetilde\nu \wedge \widetilde \nu^\top)
    \leq
    2d_\text{H}(\pi, \widetilde \pi).
\]
Moreover, we have obviously $\nu \wedge \nu^\top( |\Delta f|^4 ) \leq  \nu(|\Delta f|^4 )$ and analogously $\widetilde \nu \wedge \widetilde\nu^\top( |\Delta f|^4 ) \leq \widetilde \nu(|\Delta f|^4 )$ which concludes the proof.
\end{proof}

We now focus on linear functionals $f_v(x) = v^\top x$ with $0 \neq v \in \bbR^d$ and on $\pi = \pi_n$, $\widetilde \pi = \LAn$.
For these we show that $|1 - \frac{\Var_{\pi_n}(f_v)}{\Var_{\LAn}(f_v)}| \to 0$.

\begin{propo}\label{propo:Var_pin LAn}
Given the assumptions of Theorem \ref{theo:HRW_NonGauss} we have for any $0\neq v\in\bbR^d$ and $f_v(x) \coloneqq v^\top x$
\[
	\lim_{n\to\infty} \frac{\Var_{\pi_n}(f_v)}{\Var_{\LAn}(f_v)}
	=
	1.
\]
\end{propo}
\begin{proof}
Obviously, $\Var_{\LAn}(f_v) = v^\top C_n v = n^{-1} v^\top H_n^{-1} v$ with $H_n$ as in \eqref{equ:Hn}.
First, we consider the case that the assumptions of Theorem \ref{theo:conv_H} are satisfied, i.e., $\pi_n$ and $\LAn$ converge to a point mass.
Then, by Proposition \ref{propo:Schillings} and the assumptions we have
\[
	\Var_{\pi_n}(f_v)
	\sim
	n^{-1} v^\top H_\star^{-1} v + \mc O(n^{-2}).
\]
Since $x_n\to x_\star$ and $\nabla^2 U$ being continuous, we also have $H_n \to H_\star$ which yields
\[
    \lim_{n\to\infty} \frac{\Var_{\pi_n}(f_v)}{\Var_{\LAn}(f_v)}
    =
    \lim_{n\to\infty} \frac{n^{-1} v^\top H_\star^{-1} v}{n^{-1} v^\top H_n^{-1} v}
    =
    1.
\]
We now consider the case that the assumptions of Corollary \ref{cor:conv_H_sing} are satisfied, i.e., $\pi_n$ and $\LAn$ concentrate around a linear manifold $\mc M$.
We note that, since $\pi_0$ and $\LAn$ are Gaussian, we have for any $f_v$ that $\pi_n(f_v^4)$ and $\LAn(f_v^4)$ exist and are uniformly bounded which in combination with $d_\mathrm{H}(\pi_n, \LAn) \to 0$ yields due to \eqref{eq:Hell_Mom}
\[
   |\Var_{\pi_n}(f_v) - \Var_{\LAn}(f_v)| \to 0
\]
for any $v\in\HS$.
Next, we write
\[
    \lim_{n\to\infty} \frac{\Var_{\pi_n}(f_v)}{\Var_{\LAn}(f_v)}
    =
    1 + \lim_{n\to\infty} \frac{\Var_{\pi_n}(f_v) - \Var_{\LAn}(f_v)}{\Var_{\LAn}(f_v)}
\]
and notice that for any $v \notin \mc X$ we have $\lim_{n\to \infty} \Var_{\LAn}(f_v) > 0$.
Thus, for any $v \notin \mc X$ we have that $\frac{\Var_{\pi_n}(f_v)}{\Var_{\LAn}(f_v)}\to1$.
Let now $v \in \mc X$. 
In this case, the marginals of $\pi_n$ and $\LAn$ in this direction converge to a point mass and we can apply the reasoning from the first part of this proof to conclude that $\frac{\Var_{\pi_n}(f_v)}{\Var_{\LAn}(f_v)}\to1$.
\end{proof}

Finally, we provide another auxilliary result on the decay of the fourth moment of the directional squared jump distance regarding the Hessian-based proposal kernels.

\begin{propo}\label{propo:SJD_Proposal}
Given the assumptions of Theorem \ref{theo:HRW_NonGauss}, let $\LAn = \Nv(x_n, C_n)$ denote the Laplace approximation of $\pi_n$.
Then, for 
\begin{enumerate}
\item
the proposal kernel $P_n(x) = \Nv(x, s^2 C_n)$ with $s > 0$

\item
and the proposal kernel $P_n(x) = \Nv(x_n + \sqrt{1-s^2}(x-x_n), s^2 C_n)$ with $s \in (0,1]$
\end{enumerate}
we have for any $0\neq v \in \bbR^d$ and $f_v(x) \coloneqq v^\top x$
\[
    \sup_{n\in\bbN}
    \frac{\left( \nu_n (|\Delta f_v|^4) + \widetilde \nu_n (|\Delta f_v|^4)\right)^{1/2}}{\Var_{\LAn}(f_v)}
    < + \infty,
\]
where $\nu_n(\d x\, \d y) = P_n(x, \d y)\, \pi_n(\d x)$ and $\widetilde \nu_n(\d x\, \d y) = P_n(x, \d y)\, \LAn(\d x)$.
\end{propo}
\begin{proof}
We have $\Var_{\LAn}(f_v) = v^\top C_n v$.
For the case that $v^\top C_n v \not\to 0$, the statement follows immediately.
Thus, we consider the case, that $\pi_n$ and $\LAn$ concentrate along the direction $v$.
We first consider the case of $P_n(x) = \Nv(x,  s^2 C_n)$.
Then, since $f_v(x) = v^\top x$, we have with $\xi \sim \Nv(0, s^2 C_n)$
\[
    \nu_n (|\Delta f_v|^4)
    =
    \widetilde \nu_n (|\Delta f_v|^4)
    =
    \ev{|v^\top \xi|^4}.
\]
Now, since $v^\top \xi \sim \Nv(0, s^2 v^\top C_n v)$ and since for univariate mean normal Gaussian distributions the fourth moment coincides with three times the squared variance, we obtain for this case
\[
    \frac{\left( \nu_n (|\Delta f_v|^4) + \widetilde \nu_n (|\Delta f_v|^4)\right)^{1/2}}{\Var_{\LAn}(f_v)}
    =
    \frac{\left( 6 s^4 (v^\top C_n v)^2 \right)^{1/2}}{v^\top C_n v}
    =
    \sqrt6 s^2. 
\]
The case $P_n(x) = \Nv(x_n + \sqrt{1-s^2}(x-x_n), s^2 C_n)$ is slightly more involved.
First, We consider 
\[
    \widetilde \nu_n (|\Delta f_v|^4)
    =
    \ev{| (1-\sqrt{1-s^2}) v^\top X - v^\top\xi|^4}
\]
with $X \sim \Nv(0, C_n)$ and $\xi \sim \Nv(0, s^2 C_n)$ independently.
By introducing $Z \coloneqq (1-\sqrt{1-s^2}) X - \xi \sim \Nv(0, 2(1-\sqrt{1-s^2}) C_n)$, we get 
\[
    \widetilde \nu_n (|\Delta f_v|^4) 
    = \ev{| v^\top Z|^4}
    = 12(1-\sqrt{1-s^2})^2 \ \left(v^\top C_n v\right)^2.
\]
Moreover, we have
\begin{align*}
    \nu_n(|\Delta f_v|^4)
    & = \int_{\bbR^d} \int_{\bbR^d} \ev{|(1 - \sqrt{1-s^2}) v^\top(x-x_n) - v^\top\xi|^4} \pi_n(\d x) \ \Nv(0, s^2 C_n)(\d \xi) \\
    & \leq 2^4 \left( \int_{\bbR^d} |1 - \sqrt{1-s^2}|^4 \ |v^\top(x-x_n)|^4 \ \pi_n(\d x)  + 3 s^4 (v^\top C_n v)^2 \right).
\end{align*}
Furthermore, 
\begin{align*}
    \int_{\bbR^d} |v^\top(x - x_n)|^4\ \pi_n(\d x)
    &\leq 2^4\left( |v^\top(x_n - x_\star)|^4 + \int_{\bbR^d} |v^\top(x - x_\star)|^4\ \pi_n(\d x)
\right).
\end{align*}
Now, we know by \cite[Remark 5]{SchillingsEtAl2020} that $|v^\top (x_n - x_\star)| \in \mc O(n^{-1})$ given the assumptions of Theorem \ref{theo:laplace_method} hold on $\mathrm{span}(v)$ which is ensured by the assumptions of Theorem \ref{theo:HRW_NonGauss}.
Since for $f(x) \coloneqq |v^\top(x-x_\star)|^4$ we have $\nabla f(x_\star) = 0$ it follows by Proposition \ref{propo:Schillings} that
\begin{align*}
	\int_{\bbR^d} \left|v^\top(x-x_\star)\right|^{4}\ \pi_n(\d x) \in \mc O(n^{-2}),
\end{align*}
and, hence,
\[
    \nu_n(|\Delta f_v|^4) \in \mc O(n^{-2}).
\]
Thus, we obtain in summary for the modified pCN-proposal kernel
\begin{align*}
    \sup_{n\in\bbN} 
    \frac{\left( \nu_n (|\Delta f_v|^4) + \widetilde \nu_n (|\Delta f_v|^4)\right)^{1/2}}{\Var_{\LAn}(f_v)}
    & \leq
    \sup_{n\in\bbN}
    \frac{\left(12(1-\sqrt{1-s^2})^2 \ \left(v^\top C_n v\right)^2  + 3 s^4 (v^\top C_n v)^2 + \mc O(n^{-2}) \right)^{1/2}}
    {v^\top C_n v}\\
    & < +\infty,
\end{align*}
since $v^\top C_n v = \frac 1n v^\top H_n^{-1}v \in \mc O(n^{-1})$ due to $H_n \to H_\star>0$.
\end{proof}

\appendix

\section{Pushforward transition kernels}\label{sec:pushforward}
Let $(E,\mc E)$ be a Borel space and $(F, \mc F)$ be a measurable space. Let the mapping $T\colon E \to F$ be measurable and surjective.
If the surjectivity property is not satisfied, we can restrict the following consideration to $F$ being the image of $T$, that is, $F=T(E)$.  
Let $\pi$ be a probability measure on $(E,\mc E)$ and $K\colon E\times \mc E\to[0,1]$ be a transition kernel.
The pushforward measure of $\pi$ under $T$ is given by $T_*\pi(B) \coloneqq \pi(T^{-1}(B))$, $B\in \mc F$, and we define the \emph{pushforward transition kernel} $T_*K\colon F\times \mc F\to[0,1]$ of $K$ under $T$ by
\begin{equation}\label{equ:TK}
	T_*K(y, B)
	\coloneqq
	\ev{K(X, T^{-1}(B))\mid T(X) = y},
	\qquad
	X\sim\pi,
\end{equation}
where $y \in F,\ B\in\mc F$. We comment on this definition:
\begin{rem}
By the fact that $(E,\mc E)$ is a Borel space and \cite[Theorem~6.3]{Kallenberg2002} there exists a probability kernel $\kappa \colon F\times \mc E \to [0,1]$ as regular version of the conditional distribution of $X$ given $T(X)$ such that
\[
    \ev{K(X, T^{-1}(B))\mid T(X) = y} = \int_{E} K(x,T^{-1}(B))\, \kappa(y,\d x), \quad B\in \mc F,
\]
almost surely w.r.t. $T_*\pi$ in $y$. The right hand-side of the latter is a transition kernel and therefore we consider the right hand-side of \eqref{equ:TK} also as transition kernel. Additionally, note that for a bijective mapping $T$ we simply have
\[
	T_*K(y, B)
	=
	K(T^{-1}(y), T^{-1}(B)).
\]
\end{rem}

\begin{rem}[Algorithmic realization]
Assuming that we can sample from $K(x,\cdot)$ for any $x\in E$, drawing a sample according to $T_*K(y,\cdot)$ with $y\in F$, can be realized as follows:
\begin{enumerate}
\item
Draw $Z\sim \kappa(y,\cdot)$, where $\kappa$ denotes a regular version of the conditional distribution of $X\sim \pi$ given that $T(X)=y$, and call the result $z\in E$.
\item
Draw $X'\sim K(z, \cdot)$ and call the result $x' \in E$.
\item
Return $y' \coloneqq T(x')$.
\end{enumerate}
Thus, we need to be able to sample from the probability kernel $\kappa$, which, for bijective $T$ simplifies to sampling w.r.t. a Dirac measure at $z=T^{-1}(y)$.
\end{rem}

We have the following basic properties of $T_*K$.

\begin{propo}\label{propo:TK}
If $K\colon E \times \mc E\to [0,1]$ is $\pi$-reversible, then $T_*K$ 
is $T_*\pi$-reversible.
\end{propo}
\begin{proof}
For $A,B \in \mc F$ and $X\sim \pi$ we have
\begin{align*}
	\int_A T_*K(y, B)\ T_*\pi(\d y)
	& = \int_{A} \ev{K(X , T^{-1}(B)) \mid T(X)=y} \ T_*\pi(\d y)\\
	& = \int_{E} \boldsymbol 1_A(y)\, \ev{K(X , T^{-1}(B)) \mid T(X)=y} \ T_*\pi(\d y)\\
	& = \ev{\boldsymbol 1_A(T(X))\, \ev{K(X , T^{-1}(B)) \mid T(X)}} \\
	& = \ev{\ev{\boldsymbol 1_A(T(X))\, K(X , T^{-1}(B)) \mid T(X)}} \\
	& = \ev{\boldsymbol 1_A(T(X))\, K(X , T^{-1}(B))} \\
	& = \int_{T^{-1}(A)} K(x , T^{-1}(B))\ \pi(\d x).
\end{align*}
Analogously, we obtain
\begin{align*}
	\int_B T_*K(y, A)\ T_*\pi(\d y)
	& 
	= \int_{T^{-1}(B)} K(x , T^{-1}(A))\ \pi(\d x),
\end{align*}
such that the desired statement follows by the $\pi$-reversibility of $K$.
\end{proof}

For a $\pi$-reversible transition kernel $K$ we define the 
\emph{stationary transition measures} $\nu$ on $E\times E$ and $\nu_T$ on $F\times F$ of transition kernel $K$ and the $T_*\pi$-reversible pushforward kernel $T_*K$, respectively, as
\begin{equation}\label{equ:TK_nu}
	\nu_K(\d x\, \d x') \coloneqq K(x, \d x') \, \pi(\d x),
	\qquad
	\nu_{T_*K}(\d y\, \d y') \coloneqq T_*K(y, \d y') \, T_*\pi(\d y).
\end{equation}
Those are related as follows:
\begin{propo}\label{propo:TK_nu}
Define $\boldsymbol T\colon E\times E \to F\times F$ by $\boldsymbol T(x,x') \coloneqq (T(x), T(x'))$ and
let $K$ be reversible w.r.t. $\pi$. Then
\[
	\boldsymbol T_*\nu_K(\d y\, \d y')
	=
	T_*K(y, \d y') \, T_*\pi(\d y) = \nu_{T_*K}(\d y\, \d y').
\]
\end{propo}
\begin{proof}
For any $A, B \in \mc F$, with $X\sim\pi$, we have 
\begin{align*}
	\nu_{T_*K}(A\times B)
	& \coloneqq \int_A T_*K(y, B)\ T_*\pi(\d y)
	= \int_{A} \ev{K(X , T^{-1}(B)) \mid T(X)=y} \ T_*\pi(\d y)\\
	& = \int_{F} \boldsymbol 1_A(y) \ev{K(X , T^{-1}(B)) \mid T(X)=y} \ T_*\pi(\d y)\\
	& = \int_{F} \ev{\boldsymbol 1_A(T(X))\, K(X , T^{-1}(B)) \mid T(X)=y} \ T_*\pi(\d y)\\
	& = \ev{\ev{\boldsymbol 1_A(T(X))\, K(X , T^{-1}(B)) \mid T(X)}}\\
	& = \ev{\boldsymbol 1_A(T(X))\, K(X , T^{-1}(B))}\\
	& = \int_{T^{-1}(A)} K(x , T^{-1}(B))\, \pi(\d x)
	= \nu_K(T^{-1}(A)\times T^{-1}(B))
	= \boldsymbol T_*\nu_K(A\times B).
	\qedhere
\end{align*}
\end{proof}
We add a consequence of the former proposition.
\begin{propo}\label{propo:TK_Corr}
For a $\pi$-reversible transition kernel $K$ we have for any $f \in L^2_{T_*\pi}$ that
\[
	\int_{F\times F} (f(y)-f(y'))^2 \, \nu_{T_*K}(\d y\, \d y')
	=
	\int_{E\times E} (f\circ T(x)-f\circ T(x'))^2 \, \nu_K(\d x\, \d x')
\]
with $\nu$ and $\nu_T$ as defined in \eqref{equ:TK_nu}. 
Thus, for a Markov chains $(X_k)_{k\in\bbN}$ on $E$ with transition kernel $K$ and initial distribution $\pi$ as well as for a Markov chain $(Y_k)_{k\in\bbN}$ on $F$ with transition kernel $T_*K$ and initial distribution $T_* \pi$ we have
\begin{equation}\label{equ:TM_Corr}
	\Corr\left(f\circ T(X_k), f\circ T(X_{k+1}) \right)
	=
	\Corr\left(f(Y_k), f(Y_{k+1}) \right).
\end{equation}
\end{propo}
\begin{proof}
The first statement is an immediate consequence of Proposition~\ref{propo:TK_nu}. 
The second statement is an implication of the first one. We have for any $g \in L^2_\pi$ (by the fact that the Markov chain starts in stationarity) that
\[
	\Corr\left(g(X_k), g(X_{k+1}) \right)
	=
	1 - \frac1{2\Var_\pi(g)} \int_{E\times E} (g(x)-g(x'))^2 \, \nu_K(\d x\, \d x')
\]
and, analogously,
\begin{align*}
	\Corr\left(f(Y_k), f(Y_{k+1}) \right)
	& =
	1 - \frac1{2\Var_{T_*\pi}(f)} \int_{F\times F} (f(y)-f(y'))^2 \, \nu_{T_*K}(\d y\, \d y')\\
	& =
	1 - \frac1{2\Var_\pi(f\circ T)}\int_{E\times E} (f\circ T(x)-f\circ T(x'))^2 \, \nu_K(\d x\, \d x'),
\end{align*}
since $T_*\pi(f) = \pi(f\circ T)$ and $T_*\pi(f^2) = \pi( (f\circ T)^2 )$, thus, $\Var_{T_*\pi}(f) = \Var_{\pi}(f\circ T)$.
\end{proof}

\paragraph{Pushforwards of Metropolis--Hastings transition kernels}

Let $P\colon E\times \mc E \to [0,1]$ be a proposal transition kernel and $\pi$ be the target distribution on $(E,\mc E)$. Similarly as in \eqref{equ:TK_nu} define the measures
\[
    	\nu_P(\d x\, \d x')
	\coloneqq
	P(x, \d x') \, \pi(\d x),
	\qquad
	\nu_P^\top(\d x\, \d x')
	\coloneqq
	\nu_P(\d x'\, \d x)
\]
on $E\times E$ and assume that the density $\frac{\d \nu^\top_P}{\d \nu_P}(x,x')$ exists, i.e., we have absolute continuity $\nu_P^\top \ll \nu_P $. For any $x,x' \in E$ let $\alpha_P(x,x')$ be given by 
\[
    \alpha_P(x,x') := \min\left\{1,\frac{\d \nu^\top_P}{\d \nu_P}(x,x') \right\}.
\]
Then, the Metropolis-Hastings transition kernel $K$ on $(E,\mc E)$ with proposal kernel $P$ and acceptance probability $\alpha_P\colon E\times E \to[0,1]$ takes the form
\begin{align}\label{equ:K_MH}
	K(x, \d x')
	=
	\alpha_P(x, x') \ P(x, \d x')
	+ r(x) \delta_x(\d x'),
	\qquad
	r(x) \coloneqq  \int_E (1 -\alpha_P(x,x')) \ P(x, \d x').
\end{align}
It is well known that $K$ is reversible w.r.t. $\pi$, see \cite{Ti98}. For completeness we define the averaged acceptance rate of $K$ as 
\[
\widebar{\alpha}(K) := \int_E \int_E \alpha_P(x,x') P(x,\d x') \pi(\d x).
\]

With $T\colon E\to F$ we consider the \emph{pushforward proposal kernel $T_*P$} on $F$, i.e., 
\[
	T_*P(y, B) \coloneqq \ev{P(X , T^{-1}(B)) \mid T(X)=y}, \qquad X\sim \pi,
\]
where $y\in F$ and $B\in\mc F$. Using the notation and the result of Proposition~\ref{propo:TK_nu} we have 
\[
 \boldsymbol T_*\nu_P(\d y\, \d y')
	= 	T_*P(y, \d y') \, T_*\pi(\d y) 
	=   \nu_{T_*P}(\d y\,\d y')
\]
on $F\times F$. Define
\begin{equation}\label{eq:alpha_T}
    \alpha_{T_*P}(y, y')
    \coloneqq
    \min\left\{ 1, \frac{\d \boldsymbol T_*\nu^\top_P}{\d \boldsymbol T_*\nu_P}(y,y') \right\}, \qquad y,y'\in F,
\end{equation}
whenever it exists. If $\alpha_{T_*P}$ is well-defined, then the MH transition kernel on $F$ with proposal kernel $T_*P$ and acceptance probability $\alpha_{T_*P}$ is reversible w.r.t. $T_*\pi$. We have the following relation between $\alpha_P$ and $\alpha_{T_*P}$.  
 

\begin{propo}\label{propo:TK_alpha}
Under the assumption that the density $\frac{\d \nu^\top_P}{\d \nu_P}$ of $\nu^\top_P$ w.r.t. $\nu_P$ on $E\times E$ exists, we have that the
acceptance probability $\alpha_{T_*P}$ of \eqref{eq:alpha_T} is well-defined in the sense that $\frac{\d \boldsymbol T_*\nu^\top_P}{\d \boldsymbol T_*\nu_P}$ exists on $F\times F$. Moreover, for any $A,B \in \mc F$ holds
\[
    \int_{T^{-1}(A)}
    \int_{T^{-1}(B)}
    \alpha_P(x,x')\ 
    P(x, \d x')\ 
    \pi(\d x)
    =
    \int_A
    \int_B
    \alpha_{T_*P}(y, y')\ 
    T_*P(y,\d y')\ 
    T_*\pi(\d y).
\]
\end{propo}
\begin{proof}
By the theorem of Radon-Nicodym the existence of $h:=\frac{\d \nu^\top_P}{\d \nu_P}$ is equivalent to the absolute continuity
$\nu_P^\top \ll \nu_P$ on $E\times E$. This implies the absolute continuity $\boldsymbol T_*\nu_P^\top \ll \boldsymbol T_*\nu_P$ on $F\times F$. Namely, by $\nu_P^\top \ll \nu_P$ we have for any $A\in\mc F \otimes \mc F$ that
\[
    0 = \boldsymbol T_*\nu_P(A)
    = \nu_P(\boldsymbol T^{-1}(A))
    \quad
    \Longrightarrow
    \quad
    0
    = \nu_P^\top(\boldsymbol T^{-1}(A))
    = \boldsymbol T_*\nu^\top_P(A).
\]
Again by the theorem of Radon-Nicodym this yields the existence of $\overline{h} \coloneqq \frac{\d \boldsymbol T_*\nu^\top_P}{\d \boldsymbol T_*\nu_P}$, such that $\alpha_{T_*P}$ as given in \eqref{eq:alpha_T} is well-defined.
Furthermore, for $A,B \in \mc F$ a straightforward calculation shows 
\begin{align*}
    \int_{T^{-1}(A) \times T^{-1}(B)}
    h(x,x') 
    \nu_P(\d x \ \d x')
    & =
    \boldsymbol T_*\nu^\top_P(A\times B)
    =
    \int_{A \times B}
    \widebar h(y,y') 
    \boldsymbol T_*\nu_P(\d y \ \d y')\\
    & = 
    \int_{T^{-1}(A) \times T^{-1}(B)}
    \widebar h(T(x),T(x')) 
    \nu_P(\d x \ \d x'),
\end{align*}
which yields that $h = \widebar{h}\circ \boldsymbol T$ almost surely w.r.t. $\nu_P$. Using this we obtain
\begin{align*}
    \int_{T^{-1}(A)}
    \int_{T^{-1}(B)}
    \alpha_P(x,x')\ 
    P(x, \d x')\ 
    \pi(\d x)
    & = \int_{T^{-1}(A)}
    \int_{T^{-1}(B)}
    \min\{1, h(x,x') \} \nu_P(\d x \ \d x')\\
    & =
    \int_{T^{-1}(A) \times T^{-1}(B)}
    \min\{1, \widebar h(T(x),T(x'))\}
    \ \nu_P(\d x \ \d x')\\
    & =
    \int_{A \times B}
    \min\{1, \widebar h(y,y')\}
    \ T_*\nu_P(\d y \ \d y')\\
    & =
    \int_A
    \int_B
    \alpha_{T_*P}(y, y')\ 
    T_*P(y,\d y')\ 
    T_*\pi(\d y).
    \qedhere
\end{align*}
\end{proof}

We now show that the pushforward MH transition kernel $T_*K$ coiincides with the $T_*\pi$-reversible MH transition kernel induced by the pushforward proposal kernel $T_*P$ with acceptance probability $\alpha_{T_*P}$.

\begin{lem}\label{lem:TK_alpha}
For the pushforward transition kernel $T_*K$ of the $\pi$-reversible MH transition kernel $K$ (as defined in \eqref{equ:K_MH}) we have
\begin{equation}
    \label{eq: repres_pushforw_MH}
    T_*K(y, \d y')
    =
    \alpha_{T_*P}(y, y')\ 
    T_*P(y, \d y')
    +
    s(y)  \delta_y(\d y')
\end{equation}
with $s(y) \coloneqq  \int_F (1-\alpha_{T_*P}(y, y')) \ T_*P(y, \d y')$ and $\alpha_{T_*P}$ as given in \eqref{eq:alpha_T}.
In particular, the averaged acceptance rates of $K$ and $T_*K$ coincide, that is,
\begin{equation}
\label{eq:average_acc_coincide}
    \widebar{\alpha}(K) = \widebar{\alpha}(T_*K).
\end{equation}
\end{lem}
\begin{proof}
For any fixed $A\in \mc F$ we 
show that almost surely
\begin{equation} \label{eq:todo}
    \ev{ \int_{T^{-1}(A)} \alpha_P(X,x')\ P(X, \d x') \mid T(X)} = \int_A \alpha_{T_*P}(T(X), y')\ T_*P(T(X), \d y')
\end{equation}
where $X\sim \pi$.
This then implies by standard arguments that also 
\[
    \ev{ r(X) \delta_X(T^{-1}(A)) \mid T(X)}
    =
    s(T(X)) \ \delta_{T(X)}(A)
\]
holds almost surely and, hence, \eqref{eq: repres_pushforw_MH} follows.
In order to verify \eqref{eq:todo}, we check the definition of the conditional expectation and exploit Proposition~\ref{propo:TK_alpha}: For arbitrary $B\in \mc F$ we have
\begin{equation}  \label{eq:within_proof}
    \int_{T^{-1}(B)}
    \int_{T^{-1}(A)}
    \alpha_P(x,x')\ 
    P(x, \d x')\ 
    \pi(\d x)
    =
    \int_B
    \int_A
    \alpha_{T_*P}(y, y')\ 
    T_*P(y,\d y')\ 
    T_*\pi(\d y)
\end{equation}
and, by the fact that $X \sim \pi$, we obtain for the left hand-side
\begin{align*}
    \int_{T^{-1}(B)}
    \int_{T^{-1}(A)}
    \alpha_P(x,x')\ 
    P(x, \d x')\ 
    \pi(\d x)
    & =
    \mathbb{E}\Big [\boldsymbol 1_{B}(T(X)) 
    \mathbb{E}\Big [ \int_{T^{-1}(A)} \alpha_P(x,x')\ P(x, \d x') \mid T(X)\Big] \Big].
\end{align*}
Writing the right hand-side of \eqref{eq:within_proof} also in terms of an expectation then yields
\[
    \mathbb{E}\Big [ \boldsymbol 1_{B}(T(X)) 
    \mathbb{E}\Big [ \int_{T^{-1}(A)} \alpha_P(x,x')\ P(x, \d x') \mid T(X) \Big] \Big]
    =
    \mathbb{E}\Big [ \boldsymbol 1_{B}(T(X))
    \int_A
    \alpha_{T_*P}(T(X), y')\ 
    T_*P(T(X),\d y')
    \Big].
\]
Hence, since $B$, and, thus, $T^{-1}(B) \in \sigma(T)$ was chosen arbitrarily, by taking the obvious measurability of the right hand-side of \eqref{eq:todo} into account, we obtain the desired conditional expectation representation of \eqref{eq:todo}.

The statement of \eqref{eq:average_acc_coincide} is now a direct consequence of Proposition~\ref{propo:TK_alpha} for $A=B=F$. 
\end{proof}

If the mapping $T\colon E\to F$ satisfies further conditions, then explicit representations of $\alpha_{T_*P}$ can be determined.

\begin{propo}\label{propo:TK_MH}
Let $K\colon E\times \mc E\to[0,1]$ be as in \eqref{equ:K_MH} and $T\colon E\to F$ be a measurable function.
If there exists a measurable function $\beta \colon F \times F \to [0,1]$ such that
\[
	\beta(T(x), T(x') ) = \alpha_P(x,x') \qquad \forall x,x' \in E,
\]
then $\alpha_{T_*P} (y,y') = \beta(y,y')$ for $\nu_{T_*P}$-almost every $(y,y')\in F\times F$. 
In particular, if $T$ is bijective, then $\alpha_{T_*P}(y,y') = \alpha_P(T^{-1}(y), T^{-1}(y'))$.
\end{propo}
\begin{proof}
We have
\[
    \int_{T^{-1}(A)}
    \int_{T^{-1}(B)}
    \alpha_P(x,x')\ 
    P(x, \d x')\ 
    \pi(\d x)
    =
    \int_A
    \int_B
    \beta(y, y')\ 
    T_*P(y,\d y')\ 
    T_*\pi(\d y).
\]
Thus, by Proposition~\ref{propo:TK_alpha} and the theorem of Radon-Nicodym the first statement follows. The second statement is obvious.
\end{proof}


\paragraph{Spectral gaps of pushforward transition kernels}
We consider the \emph{pushforward Markov operator} $\mr K_T \colon L^2_{T_*\pi} \to L^2_{T_*\pi}$ associated to $T_*K$ defined by
\[
	{\rm K}_T f(y)
	\coloneqq
	\int_F
	f(y')\ T_*K(y, \d y'),
	\qquad
	f \in L^2_{T_*\pi}.
\]
For the action of $\mr K_T$ onto an $f \in L^2_{T_*\pi}$ we have the following result:

\begin{propo} \label{prop:TM_op_norm}
For ${\rm K}_T$ defined as above and any $f \in L^2_{T_*\pi}$ we have
\begin{equation}\label{equ:TM_op_f}
	\mr K_Tf(y)
	=
	\ev{\mr K (f\circ T)(X) \ | \ T(X)=y}
\end{equation}
for $T_*\pi$-almost every $y \in F$
and
\begin{equation}\label{equ:TM_op_norm}
	\|\mr K_T (f)\|_{T_*\pi}
	\leq
	\|\mr K(f \circ T)\|_{\pi}.
\end{equation}
\end{propo}
\begin{proof}
The first statement follows by
\begin{align*}
	\mr K_T f(y)
	& = \int_F f(y') \ T_*K(y, \d y')
	= \ev{ \int_F f(y') \ K(X, T^{-1}(\d y'))\mid T(X)=y}\\
	& = \ev{ \int_E f(T(x')) \ K(X, \d x')\mid T(X)=y}
	= \ev{\mr K (f\circ T)(X) \ | \ T(X)=y}.
\end{align*}
The second statement follows by an application of Jensen's inequality,
\[
	\left| \mr K_T f(y) \right|^2
	=
	\ev{\mr K (f\circ T) \ | \ T(X)=y}^2
	\leq
	\ev{|\mr K (f\circ T)(X)|^2 \ | \ T(X)=y}
\]
which yields
\begin{align*}
	\|\mr K_T (f)\|^2_{T_*\pi}
	& = \int_F \left| \mr K_T f(y) \right|^2\ T_*\pi(\d y)
	\leq 	\int_{F} \ev{|\mr K (f\circ T)(X)|^2 \ | \ T(X)=y}\ T_*\pi(\d y)\\
	& = \ev{\ev{|\mr K (f\circ T)(X)|^2 \ | \ T(X)}}
	= \ev{|\mr K (f\circ T)(X)|^2}\\
	& = \|\mr K(f\circ T)\|^2_{\pi}.
\end{align*}
\end{proof}

We now show that 
\begin{equation}\label{equ:TM_gap}
	\gapm{T_*\pi}{T_*K} \geq \gapm{\pi}{K}.
\end{equation}
To this end, we first state the following.

\begin{propo}  \label{prop:subsetballs}
For probability measure $\pi$ on $E$ we define 
\[
	B_\pi \coloneqq \left\{f \in L^2_\pi \colon \pi(f) = 0 \text{ and } \pi(f^2) = 1\right\}.
\]
Then, for $T\colon E \to F$ we have that
\[
	\{f \circ T\colon f \in B_{T_*\pi}\} \subseteq B_{\pi},
\]
where $B_{T_*\pi}\subset L^2_{T_*\pi}$ is defined correspondingly to $B_\pi$.
\end{propo}
\begin{proof}
The statement follows by
\[
	\pi(f \circ T)
	=
	\int_E f(T(x))\ \pi(\d x)
	=
	\int_F f( y ) T_*\pi(\d y)
\]
and
\[
	\pi( (f \circ T)^2 )
	=
	\int_E f^2(T(x))\ \pi(\d x)
	=
	\int_F f^2( y ) T_*\pi(\d y).\qedhere
\]
\end{proof}

\begin{theo}\label{theo:TK_gap}
Let $K\colon E \times \mc E\to [0,1]$ be a $\pi$-reversible transition kernel and $T\colon E \to F$ be a measurable mapping as above.
Then, \eqref{equ:TM_gap} holds for the Markov operators associated to $K$ and $T_*K$, respectively.
If $T$ is, furthermore, bijective, then we have equality in \eqref{equ:TM_gap}.
\end{theo}
\begin{proof}
The first statement follows by
\[
	\gapm{\pi}{K} =
	1 - \|{\rm K}\|_{B_\pi},
	\qquad
	\|\mr K\|_{B_\pi}
	\coloneqq
	\sup_{g \in B_\pi} \|\mr Kg\|_{\pi},
\]
and, analogously, $\gapm{T_*\pi}{T_*K} = 1 - \|{\rm K_T}\|_{B_{T_*\pi}}$ as well as
\begin{align*}
	\|\mr K_T\|_{B_{T_*\pi}}
	& \coloneqq \sup_{f \in B_{T_*\pi}} \|\mr K_T f\|_{T_*\pi}
	\leq \sup_{f \in B_{T_*\pi}} \|\mr K (f\circ T)\|_{\pi}\\
	& = \sup_{f\circ T \in B_\pi \colon f \in B_{T_*\pi}} \|\mr K (f\circ T)\|_{\pi} \leq \|\mr K\|_{B_\pi},
\end{align*}
using Proposition~\ref{prop:TM_op_norm} and Proposition~\ref{prop:subsetballs}.
If $T$ is bijective, then we can use the same argumentation, but starting from $T_*K$ as the given transition kernel and considering $K$ as the pushforward transition kernel of $T_*K$ under $T^{-1}$.
This yields $\gapm{T_*\pi}{T_*K} \leq \gapm{\pi}{K}$ and, thus, the second statement.
\end{proof}

\bibliographystyle{plain}
\bibliography{literature}
\end{document}

%% file: macros.tex
\DeclareMathAlphabet{\mathbf}{OT1}{cmr}{bx}{it}

\newcommand{\supp}{\mathrm{supp}\,}

\newcommand{\mc}{\mathcal}

\newcommand{\mr}{\mathrm}

\DeclareMathAlphabet{\mathbf}{OT1}{cmr}{bx}{it}

\DeclareMathOperator{\diag}{diag}

\newcommand{\argmin}{\operatornamewithlimits{argmin}}

\newcommand{\essinf}{\operatornamewithlimits{ess\ inf}}

\newcommand{\K}{\mathrm K}

\newcommand{\bbN}{\mathbb{N}}
\newcommand{\bbR}{\mathbb{R}}

\newcommand{\bbP}{\mathbb{P}}

\DeclareMathOperator{\e}{\mathrm e}
\newcommand{\ev}[1]{ {\boldsymbol{\mathbb  E}} \left[  #1 \right]}

\renewcommand{\d}{\mathrm{d}}
\DeclareMathOperator{\Cov}{Cov}
\DeclareMathOperator{\Corr}{Corr}

\DeclareMathOperator{\Var}{Var}

\theoremstyle{definition}
\newtheorem{defi}{Definition}
\newtheorem{exam}[defi]{Example}
\newtheorem{rem}[defi]{Remark}
\newtheorem{assum}[defi]{Assumption}
\newtheorem{propo}[defi]{Proposition}
\newtheorem{theo}[defi]{Theorem}
\newtheorem{lem}[defi]{Lemma}